\documentclass{article}

\usepackage{hyperref}
\usepackage[accepted]{icml2024}
\usepackage[T1]{fontenc}
\usepackage[utf8]{inputenc}
\usepackage[english]{babel}
\usepackage{enumitem}
\setlist[enumerate]{itemsep=0.1mm, parsep=0mm}
\usepackage{microtype}
\usepackage{graphicx}
\usepackage{subcaption}
\usepackage{listings}
\usepackage{xcolor}
\usepackage{dirtytalk}
\usepackage{amsmath}
\usepackage{amssymb}
\usepackage{amsthm}
\usepackage{mathtools}
\usepackage{thmtools, thm-restate}
\usepackage{mathabx}
\usepackage{nicefrac}
\usepackage[capitalize,noabbrev]{cleveref}
\usepackage{bbm}

\declaretheoremstyle[
 headfont=\bfseries
]{customstyle}

\declaretheorem[style=customstyle]{definition}

\newtheorem{alemma}{Lemma}[section]

\theoremstyle{definition}
\newtheorem{remark}{Remark}

\newcommand{\mech}{\mathcal{M}}
\newcommand{\epsdel}{(\varepsilon, \delta)}
\newcommand{\epsdeleps}{(\varepsilon, \delta(\varepsilon))}
\newcommand{\deleps}{\delta(\varepsilon)}
\newcommand{\prob}[1]{\mathrm{Pr}\left( #1 \right)}
\newcommand{\re}{\mathrm{e}}
\newcommand{\rd}{\mathrm{d}}
\newcommand{\eps}{\varepsilon}

\newcommand{\textand}{\; \text{and} \;}

\newcommand{\aphi}{\alpha_{\phi}}
\newcommand{\bphi}{\beta_{\phi}}

\newcommand{\bE}{\mathbb{E}}
\newcommand{\wbx}{\widebar{X}}
\newcommand{\wby}{\widebar{Y}}
\newcommand{\dtv}{\mathsf{TV}}

\newcommand{\rmin}{R_{\min}}
\newcommand{\forallalpha}{\; \forall \alpha \in [0,1]}
\newcommand{\forallpi}{\; \forall \pi \in [0,1]}
\newcommand{\rast}{R^{\ast}}
\newcommand{\mpp}{\mech_{\mathrm{PP}}}
\newcommand{\mbnp}{\mech_{\mathrm{BNP}}}
\newcommand{\renyidiv}[3]{\mathsf{D}_{#1}(#2 \mkern-5mu \parallel \mkern-5mu #3)}

\newcommand{\fD}{\mathfrak{D}}
\newcommand{\deltadiv}[2]{\Delta(#1 \mkern-5mu \parallel \mkern-5mu #2)}
\newcommand{\deltatwosided}[2]{\Delta^{\leftrightarrow}(#1 \mkern-5mu \parallel \mkern-5mu #2)}
\newcommand{\mypar}[1]{\textbf{#1} \,}

\icmltitlerunning{Beyond the Calibration Point: Mechanism Comparison in Differential Privacy}

\begin{document}

\twocolumn[
\icmltitle{Beyond the Calibration Point: Mechanism Comparison in Differential Privacy}

\icmlsetsymbol{equal}{*}

\begin{icmlauthorlist}
\icmlauthor{Georgios Kaissis}{equal,tum}
\icmlauthor{Stefan Kolek}{equal,lmu}
\icmlauthor{Borja Balle}{gdm}
\icmlauthor{Jamie Hayes}{gdm}
\icmlauthor{Daniel Rueckert}{tum}
\end{icmlauthorlist}

\icmlaffiliation{tum}{AI in Healthcare and Medicine and Institute of Radiology, Technical University of Munich, Germany}
\icmlaffiliation{lmu}{Mathematical Foundations of AI, LMU Munich}
\icmlaffiliation{gdm}{Google DeepMind}

\icmlcorrespondingauthor{Georgios Kaissis}{g.kaissis@tum.de}

\icmlkeywords{Machine Learning, ICML}

\vskip 0.3in
]



% \printAffiliationsAndNotice{} % leave blank if no need to mention equal contribution
\printAffiliationsAndNotice{\icmlEqualContribution} % otherwise use the standard text.

\begin{abstract}
In differentially private (DP) machine learning, the privacy guarantees of DP mechanisms are often reported and compared on the basis of a single $\epsdel$-pair. 
This practice overlooks that DP guarantees can vary substantially \emph{even between mechanisms sharing a given $\epsdel$}, and potentially introduces privacy vulnerabilities which can remain undetected.
This motivates the need for robust, rigorous methods for comparing DP guarantees in such cases.
Here, we introduce the $\Delta$-divergence between mechanisms which quantifies the worst-case excess privacy vulnerability of choosing one mechanism over another in terms of $\epsdel$, $f$-DP and in terms of a newly presented Bayesian interpretation.
Moreover, as a generalisation of the Blackwell theorem, it is endowed with strong decision-theoretic foundations.
Through application examples, we show that our techniques can facilitate informed decision-making and reveal gaps in the current understanding of privacy risks, as current practices in DP-SGD often result in choosing mechanisms with high excess privacy vulnerabilities.
\end{abstract}

\section{Introduction}
Protecting private information in machine learning (ML) workflows involving sensitive data is of paramount importance.
Differential Privacy (DP) has emerged as the preferred method for providing rigorous and verifiable privacy guarantees, quantifiable by a \emph{privacy budget}. 
This represents the privacy loss incurred by publicly releasing data that has been processed by a system using DP, e.g.\@ when a deep learning model is trained on sensitive data using DP stochastic gradient descent (DP-SGD, \citep{abadi2016deep}).
In principle, workflows utilising DP can offer strong protection against specific attacks, such as membership inference (MIA) and data reconstruction attacks. 
However, the proper application of DP to defend against such threats relies on a correct understanding of the quantitative aspects of privacy protection, which are expressed differently under the various DP interpretations.
For instance, in approximate DP, the privacy budget is quantified using two parameters $\epsdel$.
Most relevant DP mechanisms, e.g.\@ the subsampled Gaussian mechanism (SGM) typically used in DP-SGD, satisfy DP across a \emph{continuum} of $\epsdeleps$-values rather than a single $\epsdel$ tuple.
For these mechanisms, $\delta$ is a function of $\eps$, represented as the \emph{privacy profile} \citep{balle2020privacy}. 
An equivalent (\emph{dual}) functional view is expressed by the trade-off function in $f$-DP \citep{dong2022gaussian}.

However, despite the fact that the DP guarantee of such mechanisms can only be characterised by a collection of $\epsdel$-values, it is common practice in literature to calibrate against and report \emph{a single} $\epsdel$-pair to express the privacy guarantee of a DP mechanism \citep{abadi2016deep, papernot2021tempered, de2022unlocking}.
This highlights a potential misconception that such a single pair is sufficient to fully characterise or compare DP guarantees.
This assumption is not generally true, as mechanisms can conform to the same $\epsdel$-values but still differ significantly, as seen in \cref{fig:teaser}.
In other words: \emph{two DP mechanisms can be calibrated to share an $\epsdel$-guarantee while offering substantially different privacy protections}.

\begin{figure}[ht]
  \centering
  \includegraphics[width=0.9\columnwidth]{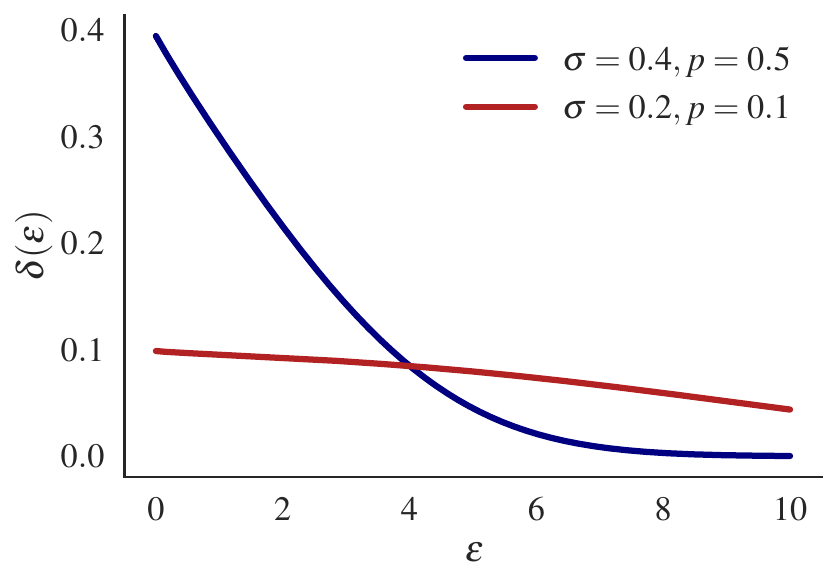}
  \caption{
  Privacy profiles of two SGMs with different noise scales $\sigma$ and sampling rates $p$.
  Both satisfy $(4.00, 0.08)$-DP, but offer otherwise different levels of privacy protection.
  }
  \label{fig:teaser}
\end{figure}

This leads us to ask whether interpreting and/or comparing the privacy guarantees of DP mechanisms based on their behaviours at a single $\epsdel$-tuple can lead to privacy vulnerabilities.
An affirmative answer is suggested by the recent work of \citet{hayes2023bounding} on reconstruction attacks. 
Therein, the authors demonstrate that calibrating two SGMs with different parameters to meet the same $\epsdel$-guarantee as shown above results in disparate effectiveness against reconstruction attacks.
In practice, this can occur when the user simultaneously increases the sampling rate (e.g.\@ to utilise all available GPU memory) and the noise scale in an attempt to maintain the same $\epsdel$-DP guarantee. 
In reality, the privacy guarantee has been changed \emph{everywhere except the calibration point} (i.e. the $\epsdel$-tuple in question), weakening the model's protection against data reconstruction attacks.
Similar evidence was presented by \citet{lokna2023group}, where it was shown that a single $\epsdel$-pair is insufficient to fully characterise a mechanism's protection against MIA. 
Both examples illustrate that differences between DP guarantees which remain undetected by only considering a single $\epsdel$-pair can lead to privacy hazards.

This reflects an unmet requirement for tools to quantitatively compare the privacy guarantees offered by DP mechanisms in a principled manner.
Most existing techniques for comparing DP guarantees either rely on summarisation into a single scalar (which can discard information), on average-case metrics or on assumptions, thus lacking the required generality. 
The arguably most theoretically rigorous mechanism comparison technique relies on the so-called \emph{Blackwell theorem}, which allows for comparing the privacy guarantees in a strong, decision-theoretic sense. 
However, the Blackwell theorem is exclusively applicable to the special case in which the privacy guarantees of two mechanisms coincide nowhere, i.e.\@ when their trade-off functions/privacy profiles \emph{never cross}, excluding, among others, DP-SGD, as shown above.
To thus extend rigorous mechanism comparisons to this important setting, a set of novel techniques is required, which our work introduces through the following contributions.

\mypar{Contributions}
To enable principled comparisons between mechanism whose privacy guarantees coincide at a single point but differ elsewhere, we generalise Blackwell's theorem by introducing an \emph{approximate} ordering between DP mechanisms.
This ordering, which we express through the newly presented \emph{$\Delta$-divergence} between mechanisms, quantifies the worst-case increase in privacy vulnerability incurred by choosing one mechanism over another in terms of hypothesis testing errors, $\deleps$, and in terms of a novel \textit{Bayes error interpretation}. 
The latter is a probabilistic extension of the hypothesis testing interpretation of DP and allows for principled reasoning over the capabilities of DP adversaries.
In addition, we analyse the evolution of approximate comparisons into universal comparisons under composition, yielding insights into the privacy dynamics of algorithms like DP-SGD.
Finally, we experimentally show how our techniques can facilitate a more granular privacy analysis of private ML workflows, and pinpoint vulnerabilities which remain undetected by only focusing on a single $\epsdel$-pair.

\mypar{Related Work}
Blackwell's theorem \citep{blackwell1953equivalent} originates in the theory of comparisons between information structures called \emph{statistical experiments}, and describes conditions under which one statistical experiment is universally more informative than another. 
Blackwell's framework was later expanded by \citet{le1964sufficiency, torgersen1991comparison}, and we refer to the latter for a comprehensive overview of the field.
The equivalence between a subclass of statistical experiments (binary experiments) and the decision problem faced by the MIA adversary led \citet{dong2022gaussian} to leverage the Blackwell theorem to provide conditions under which one DP mechanism is \emph{universally} more private than another. 
This limits mechanism comparisons to the special case when the mechanisms' trade-off functions (or privacy profiles \citep{balle2020privacy}) never cross.
However, as demonstrated above, crossing trade-off functions or privacy profiles are \emph{not the exception but the norm}; however, no specific tools to compare privacy guarantees in this case are introduced by \citet{dong2022gaussian}. 

As discussed above, privacy guarantees have so far often be compared using metrics like attack accuracy or area under the trade-off curve (see \citet{carlini2022membership} for a list of works). 
Besides summarising the privacy guarantee into a single scalar (thus discarding much of the information about the DP mechanism contained in the privacy profile or trade-off function), such metrics model the average case instead of the desirable worst case, rendering them sub-optimal for DP applications.
To remedy this, \citet{carlini2022membership} proposed comparing attack performance at a \say{low} Type-I error. 
However, this method requires an arbitrary assumption about the correct choice of a \say{low} Type-I error rather than considering the entire potential operating range of an adversary, thereby also discarding information.
Moreover, absent a universally agreed upon standard of what a correct choice of Type-I error is, this could incentivise the reporting of research results at a Type-I error which is \say{cherry-picked} to e.g.\@ emphasise the benefits of a newly introduced MIA, i.e.\@ \textit{$p$-hacking} \cite{Wasserstein2016}. 
 
The concept of comparing the privacy properties or \say{leakage} of systems has precedent in the \textit{quantitative information flow} literature \citep{alvim2020science}; for example, the minimax Bayes error in \cref{eq:minimax-bayes-error} is related to the Bayes security metric introduced in \citet{chatzikokolakis2023bayes}.

\mypar{Notation and Background}
Here, we briefly introduce the notation and relevant concepts used throughout the paper for readers with technical familiarity with DP terminology.
A detailed background discussion introducing all following concepts can be found in \cref{sec:extended-Background}.
We will denote DP mechanisms by $\mech:(P,Q)$, where $(P,Q)$ denote the \emph{tightly dominating pair} of probability distributions which characterise the mechanism as described in \citet{zhu2022optimal}, and will assume that $P$ and $Q$ are mutually absolutely continuous.
The Likelihood Ratios (LRs) will be denoted $\wbx = \nicefrac{Q(\omega)}{P(\omega)}, \omega \sim P$ and $\wby = \nicefrac{Q(\omega)}{P(\omega)}, \omega \sim Q$ for a mechanism outcome $\omega$, where $\sim$ denotes sampling, and the Privacy Loss Random Variables (PLRVs) will be denoted $X = \log(\wbx)$ and $Y = \log(\wby)$.
We will denote the trade-off function \citep{dong2022gaussian} corresponding to $\mech$ by $f : \alpha \mapsto \beta(\alpha)$, where $(\alpha, \beta(\alpha))$ are the Type-I/II errors of the most powerful test between $P$ and $Q$ with null hypothesis $H_0 : \omega \sim P$ and alternative hypothesis $H_1 : \omega \sim Q$, and $\alpha$ is fixed by the adversary.
We will assume without loss of generality that $f$ is symmetric (thereby omitting the dominating pair $(Q,P)$), and defined on $\mathbb{R}$ with $f(x) = 1, x<0$ and $f(x) = 0, x>1$.

The privacy profile \cite{balle2020privacy, gopi2021numerical} of $\mech$ will be denoted by $\deleps$, while the $N$-fold self-composition of $\mech$ (as is usually practised in DP-SGD \citep{abadi2016deep}) will be denoted by $\mech^{\otimes N}$.
We will moreover denote the total variation distance between $P$ and $Q$ by $\dtv(P,Q) = \max_{\alpha} (1 - \alpha - f(\alpha)) = \mathsf{Adv}$, where $\mathsf{Adv}$ is the MIA advantage \cite{yeom2018privacy}, and the Rényi divergence of order $t$ of $P$ to $Q$ by $\renyidiv{t}{P}{Q}$ \citep{mironov2017renyi}.
The party employing a DP mechanism to protect privacy will be referred to as the \emph{analyst} or \emph{defender}.

\section{A Bayesian Interpretation of $f$-DP}
We begin by introducing a novel interpretation of $f$-DP based on the \emph{minimum Bayes error} of a MIA adversary.
While $f$-DP characterises mechanisms through their trade-off between hypothesis testing errors, our interpretation enriches this characterisation by incorporating the adversary's \emph{prior knowledge} (i.e. auxiliary information).
As will become evident below, this allows for incorporating probabilistic reasoning over the adversary and facilitates intuitive operational interpretations of mechanism comparisons, while preserving the same information as $f$-DP.

Suppose that a Bayesian adversary assigns a prior probability $\pi$ to the decision \say{reject $H_0$}.
Considering that the adversary's goal is a successful MIA on a specific \emph{challenge example}, $H_0$ is synonymous with the hypothesis \say{the mechanism outcome was generated from the database which does not contain the challenge example}.
Thus, the prior on rejecting $H_0$ expresses the prior belief that the challenge example is actually part of the database (i.e.\@ a prior probability of positive membership).
For example, in privacy auditing (where the analyst assumes the role of the adversary), $\pi$ corresponds to the probability of including the challenge example (also called \say{canary}) in the database which is attacked \citep{carlini2022membership, nasr2023tight}.

From the trade-off function, the Bayes error $R$ at a prior $\pi$ can be obtained as follows:
\begin{equation}\label{eq:bayes-risk-definition}
  R(\pi) = \pi \alpha + (1-\pi) f(\alpha),
\end{equation}
where it is implied that the adversary fixes a level of Type I error $\alpha$.
The \emph{minimum Bayes error function} is derived from the above by minimising over the trade-off between Type I and Type II errors:
\begin{equation}
  \rmin(\pi) = \min_{\alpha}\left(\pi \alpha + (1-\pi)f(\alpha)\right).
\end{equation}
We will refer to $\rmin$ as just the \emph{Bayes error function} for short.
$\rmin$ is continuous, concave, maps $[0,1] \to [0, \nicefrac{1}{2}]$, satisfies $\rmin(0) = \rmin(1) = 0$, and $\rmin(\pi) \leq \min\{ \pi, 1-\pi\}$.
The \emph{minimax Bayes error} $\rast$ is the maximum of $\rmin$ over all values of $\pi \in [0,1]$:
\begin{equation}\label{eq:minimax-bayes-error}
  \rast = \max_{\pi} \rmin(\pi).
\end{equation}
$\rast$ is realised at $\pi=\nicefrac{1}{2}$ since $f$ is assumed symmetric.

$\rmin$ is a lossless representation of the mechanism's privacy properties as $f$ can be reconstructed from $\rmin$ as follows: 
\begin{equation}
  f(\alpha) = \max_{0\leq \pi<1} \left(-\frac{\pi}{1-\pi}\alpha + \frac{R_{\min}(\pi)}{1-\pi}\right).
\end{equation}
For examples of $\rmin$, see \cref{fig:comparison-demo} and \cref{fig:tradeoff-bayes} in the Appendix.

\section{Blackwell Comparisons}
\subsection{Universal Blackwell Dominance}
As stated above, the Blackwell theorem states equivalent conditions under which a mechanism $\mech$ is \emph{universally more informative/less private} than a mechanism $\widetilde \mech$, denoted $\mech \succeq \widetilde \mech$ from now on.
For completeness, we briefly re-state these conditions here, and extend them to include our novel Bayes error interpretation.

\begin{restatable}{theorem}{blackwelltheorem}\label{thm:blackwell}
  The following statements are equivalent:
  \begin{enumerate}
    \item $\forall \alpha \in[0,1]: f(\alpha) \leq \widetilde f(\alpha)$;
    \item $\forall \eps \in \mathbb{R}: \deleps \geq \widetilde \delta(\eps) $;
    \item $\forall \pi \in [0,1]: \rmin(\pi) \leq \widetilde R_{\min}(\pi)$.
  \end{enumerate}
\end{restatable}

The proofs of clause (1) and (2) can be found in Sections 2.3 and 2.4 of \citet{dong2022gaussian}, while the proof of (3) and all following theoretical results can be found in \cref{subsec:proofs}.

If any of the above conditions hold, we write $\mech \succeq \widetilde \mech$ and say that \emph{$\mech$ Blackwell dominates $\widetilde \mech$}. 
Note the lack of a clause related to Rényi DP (RDP), which is a consequence of the fact that, while $\mech \succeq \widetilde \mech$ implies that $\renyidiv{t}{P}{Q} \geq \renyidiv{t}{\widetilde P}{\widetilde Q}$, for all $t\geq 1$, the reverse does not hold in general \citep{dong2022gaussian}.
RDP is thus a \emph{generally weaker basis of comparison} between DP mechanisms.

The relation $\succeq$ induces a partial order on the space of DP mechanisms and expresses a strong condition, as it implies that the dominating mechanism is more useful for \emph{any} downstream task, benign (e.g.\@ training an ML model) or malicious (e.g.\@ privacy attacks) \citep{dong2022gaussian, torgersen1991comparison}.
Note that \cref{thm:blackwell} is inapplicable when the trade-off functions, privacy profiles or Bayes error functions cross. 
Addressing this issue is the topic of the rest of the paper.

\subsection{Approximate Blackwell Dominance}\label{sec:approx-comparisons}
As discussed above, Blackwell dominance expresses that choosing the dominated mechanism is, in a universal sense, a better choice in terms of privacy protection.
In other words, an analyst choosing the dominated mechanism would \emph{never regret} this choice from a privacy perspective.
However, more frequently, the choice between mechanisms is equivocal because their privacy guarantees coincide at the calibration point, but differ elsewhere.
They thus offer disparate protection against different adversaries, meaning that no choice fully eliminates potential regret in terms of privacy vulnerability.
A natural decision strategy under the principle of DP to protect against the worst case is to choose the mechanism which minimises the \emph{worst-case regret} in terms of privacy vulnerability.
To formalise this strategy, we next introduce a relaxation of the Blackwell theorem. 
Similar to how approximate DP relaxes pure DP, we term comparisons using this relaxation \emph{approximate Blackwell comparisons}.\footnote{A related term in the experimental comparisons literature is \say{deficiencies} \citep{le1964sufficiency}.}

To motivate this formalisation within the DP threat model, suppose that an analyst must choose between $\mech$ and $\widetilde \mech$, however they cannot unequivocally decide between them because \emph{neither mechanism is universally more or less vulnerable to MIA}.
To express \say{how close} the analyst is to being able to choose unequivocally between the mechanisms (i.e.\@ to Blackwell dominance being restored), we determine the smallest shift $\kappa \geq 0$ which suffices to move $f$ below and to the left of $\widetilde f$ such that \cref{thm:blackwell} kicks in and $\mech \succeq \widetilde \mech$, as shown in \cref{fig:divergence-illustration}.
\begin{figure}[ht]
  \centering
  \includegraphics[width=0.9\columnwidth]{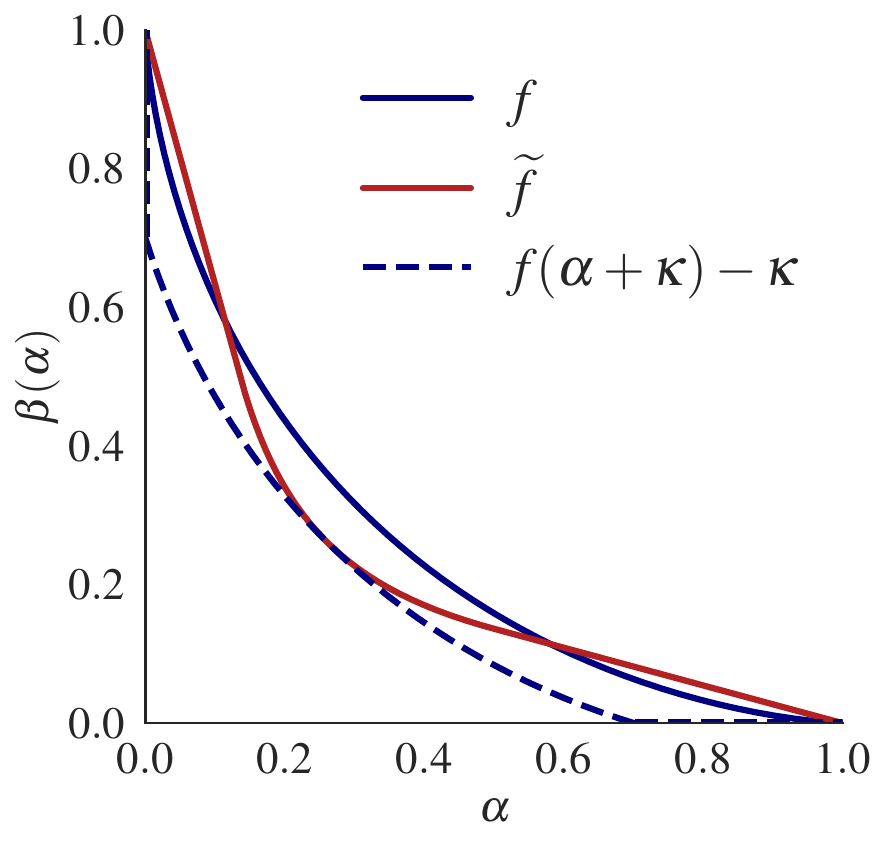}
  \caption{
  The trade-off functions of a Gaussian ($f$, blue) and a Laplace ($\widetilde f$, red) mechanism cross, therefore neither mechanism universally Blackwell dominates.
  Shifting $f$ left and down by $\kappa$ yields $f(\alpha) \leq \widetilde f(\alpha)$ for all $\alpha \in [0,1]$ (dashed blue), and restores universal Blackwell dominance.
  }
  \label{fig:divergence-illustration}
\end{figure}

\begin{definition}
The \emph{$\Delta$-divergence} of $\mech$ to $\widetilde \mech$ is given by
  \begin{equation*}
     \deltadiv{\mech}{\widetilde \mech} = \inf\{ \kappa\geq 0 \mid \forall \alpha: \; f(\alpha + \kappa) - \kappa \leq \widetilde f(\alpha) \}.
  \end{equation*}
\end{definition}
This allows us to define approximate Blackwell dominance:
\begin{definition}
  If $\deltadiv{\mech}{\widetilde \mech} \leq \fD$, we say that $\mech$ $\fD$-approximately dominates $\widetilde \mech$, denoted $\mech \succeq_{\fD} \widetilde \mech$.
\end{definition}

The next theorem formally states equivalent criteria for approximate Blackwell dominance:

\begin{restatable}{theorem}{approxcomp}\label{thm:approx-comparison}
  The following are equivalent to $\mech \succeq_{\fD}\widetilde \mech$:
  \begin{enumerate}
    \item $\forallalpha: f(\alpha + \fD) - \fD \leq \widetilde f(\alpha)$;
    \item $\; \forall \eps \in \mathbb{R}: \delta(\eps) + \fD\cdot(1+\re^{\eps}) \geq \widetilde \delta(\eps)$;
    \item $\forallpi: \rmin(\pi) - \widetilde R_{\min}(\pi) \leq \fD$.
  \end{enumerate}
\end{restatable}
The proof relies on fundamental properties of trade-off functions, of the convex conjugate and its order-reversing property and on the lossless conversion between trade-off function and Bayes error function.

Intuitively, when $\fD$ is very small, the clauses of \cref{thm:approx-comparison} are \say{approximate} versions of the corresponding clauses of \cref{thm:blackwell}.
In particular, $\fD$ represents an upper bound on the excess vulnerability of $\widetilde \mech$ at any level $\alpha$, choice of $\eps$ or prior $\pi$.
The computation of $\deltadiv{\mech}{\widetilde \mech}$ is most naturally expressed through the Bayes error functions: 
\begin{restatable}{corollary}{deltadivbayes}\label{lemma:deltadiv-via-bayes}
$\deltadiv{\mech}{\widetilde \mech} = \max_{\pi}(\rmin(\pi)-\widetilde R_{\min}(\pi))$.
\end{restatable}
The $\Delta$-divergence can be computed numerically through grid discretisation with $N$ points (i.e.\@ to tolerance $\nicefrac{1}{N}$) in $\mathcal{O}(N)$ time, and requires only oracle access to a function implementing the trade-off functions of the mechanisms.
An example is provided in \cref{sec:divergence-computation}.

Moreover, \cref{lemma:deltadiv-via-bayes} admits the following interpretation: 
\begin{quote}
$\deltadiv{\mech}{\widetilde \mech}$ \emph{expresses the worst-case regret of an analyst choosing to employ $\widetilde \mech$ instead of $\mech$, whereby regret is expressed in terms of the adversary's decrease in minimum Bayes error.}
\end{quote}
We consider this connection between Bayesian decision theory and DP the most natural interpretation of our results. 

\subsection{Metrising the Space of DP Mechanisms}\label{sec:metrising}
After introducing tools for establishing a ranking between DP mechanisms in the preceding sections, we here show that the $\Delta$-divergence can actually be used to define a \emph{metric} on the space of DP mechanisms. 
In the sequel, we will say that two mechanisms are \textit{equal} and write $\mech = \widetilde \mech$ if and only if their trade-off functions, privacy profiles and Bayes error functions are equal.
For a formal discussion on this choice of terminology, see \cref{remark:equivalence} in the Appendix.
Moreover, we define the following extension of the $\Delta$-divergence:
\begin{definition}[Symmetrised $\Delta$-divergence $\Delta^{\leftrightarrow}$]
  \begin{equation}\label{eq:delta-twosided}
    \deltatwosided{\mech}{\widetilde \mech} = \max \left \{ \deltadiv{\mech}{\widetilde \mech}, \deltadiv{\widetilde \mech}{\mech} \right \}.
  \end{equation} 
\end{definition}
Using \cref{lemma:deltadiv-via-bayes}, $\Delta^{\leftrightarrow}$ can be written as:
\begin{equation}
  \deltatwosided{\mech}{\widetilde \mech} = \lVert \rmin(\pi) - \widetilde R_{\min}(\pi) \lVert_{\infty}.
\end{equation}
In terms of the trade-off functions, the following holds:
\begin{restatable}{lemma}{levy} \label{lemma:levy}
  Let $\Delta^{\leftrightarrow} = \deltatwosided{\mech}{\widetilde \mech}$.
  Then it holds that:
  \begin{equation}\label{eq:sandwich-inequality-for-twosided}
    f(\alpha + \Delta^{\leftrightarrow}) - \Delta^{\leftrightarrow} \leq \widetilde f(\alpha) \leq f(\alpha - \Delta^{\leftrightarrow}) + \Delta^{\leftrightarrow}.
  \end{equation} 
\end{restatable}
This substantiates that $\deltatwosided{\mech}{\widetilde \mech}=0$ is equivalent to the equality of the trade-off functions, and thus of the privacy profiles and Bayes error functions.

The similarity of \cref{eq:sandwich-inequality-for-twosided} to the Lévy distance is not coincidental, and it is shown in \cref{sec:levy-equivalence} that, by considering the trade-off function as a CDF (via $f(1-\alpha)$), $\Delta^{\leftrightarrow}$ exactly plays the role of the Lévy distance. 
Similarly to how the Lévy distance metrises the weak convergence of random variables, $\Delta^{\leftrightarrow}$ metrises the space of DP mechanisms:
\begin{restatable}{corollary}{pseudometric}
  $\Delta^{\leftrightarrow}$ is a metric.
\end{restatable}
Note that this implies that $\deltatwosided{\mech}{\widetilde \mech}>0$ unless the mechanisms have identical privacy profiles, trade-off functions or Bayes error functions, underscoring that sharing a single $\epsdel$-guarantee is an insufficient condition for stating that mechanisms provide equal protection.

\subsection{Comparisons with Extremal Mechanisms}\label{sec:extremal-comparisons}
Next, we use the $\Delta$-divergence to interpret comparisons with two \say{extremal} reference mechanisms: the \emph{blatantly non-private} (totally informative) mechanism $\mbnp$ and the \emph{perfectly private} (totally non-informative) mechanism $\mpp$.
These two mechanisms represent the \say{extremes} of the privacy/information spectrum.

For this purpose, we define for $\mbnp$: $f_{\mathrm{BNP}}(\alpha) = 0$, $R_{\min}^{\mathrm{BNP}}(\pi) = 0$, and $\delta_{\mathrm{BNP}}(\eps)=1$.
Moreover, we define for $\mpp$: $f_{\mathrm{PP}}(\alpha)=1-\alpha$, $R_{\min}^{\mathrm{PP}}(\pi) = \min\{\pi, 1-\pi\}$, and $\delta_{\mathrm{PP}}(\eps)=0$.
The next lemma establishes the \say{extremeness}:
\begin{restatable}{lemma}{extremeness}
  $\mech \succeq \mpp$ and $\mbnp \succeq \mech$ for any $\mech$.
\end{restatable}
We can thus compute a \say{divergence from perfect privacy} $\deltadiv{\mpp}{\mech}$, and a \say{divergence to blatant non-privacy} $\deltadiv{\mech}{\mbnp}$.
Both have familiar operational interpretations in terms of quantities from the field of DP:
\begin{restatable}{lemma}{perfectprivacy}
  It holds that $\deltadiv{\mpp}{\mech} = \nicefrac{1}{2}\dtv(P,Q) =  \nicefrac{1}{2}\mathsf{Adv} = \nicefrac{1}{2} \delta(0)$.
\end{restatable}
This conforms to the intuition that, the \say{further} the mechanism is from perfect privacy, the higher the adversary's MIA advantage can be.
\begin{restatable}{lemma}{blatantnonprivacy}\label{lemma:blatant-non-privacy}
  It holds that $\deltadiv{\mech}{\mbnp} = \rast = \alpha^{\ast}$, where $\rast$ is the minimax Bayes error and $\alpha^{\ast}$ the fixed point of the trade-off function of $\mech$.
\end{restatable}
Recall that $\rast$ is the error rate of an \say{uninformed} adversary ($\pi=0.5$, compare \cref{fig:bayes-from-tradeoff}), whereas $\alpha^{\ast}$ is the point on the trade-off curve closest to the origin, i.e.\@ to $(\alpha,f(\alpha))=(0,0)$ (see \cref{fig:tradeoff-from-bayes}).
When either point coincides with the origin, the mechanism is blatantly non-private.
Moreover, the following holds for any mechanism:
\begin{restatable}{lemma}{complementaryerrors}\label{lemma:complementary-errors} 
  $\deltadiv{\mpp}{\mech}+ \deltadiv{\mech}{\mbnp} = 0.5$.
\end{restatable}
The results of \cref{sec:metrising} and \cref{sec:extremal-comparisons} lead to the following conclusions:
On one hand, the metric $\Delta^{\leftrightarrow}$ can be used to measure a notion of \say{informational distance} even between \textit{completely different} mechanisms (e.g.\@ Randomised Response and DP-SGD). 
Additionally, the space of DP mechanisms is a \textit{bounded partially ordered set} with a maximal ($\mbnp$) and a minimal ($\mpp$) bound, and \textit{any} DP mechanism can be placed on the information spectrum between them. 

While not discussed in detail here, we note that this set is also a \textit{lattice} \citep[Theorem 10]{blackwell1953equivalent}.
Intuitively, the significance of this lattice structure is that it guarantees a globally coherent ordering of privacy guarantees. 
For any two DP mechanisms $\mathcal{M}_1$ and $\mathcal{M}_2$, it ensures the existence of a unique \textit{greatest lower bound} (meet, $\mech_1 \wedge \mech_2$) representing the least private mechanism that is still more private than both, and a unique \textit{least upper bound} (join, $\mech_1 \vee \mech_2$) representing the most private mechanism that is still less private than both. 
Note that $\forall \mech: \mathcal{M} \wedge \mpp = \mpp$ and $\forall \mech: \mathcal{M} \vee \mbnp = \mbnp$ as a consequence of $\mpp$'s and $\mbnp$'s extremal properties.

\section{Emergent Blackwell Dominance} \label{sec:emergent-blackwell}
We next study the interplay of mechanism comparisons and composition. 
The fact that $\mech \succeq \widetilde \mech$ implies $\mech^{\otimes N} \succeq \widetilde\mech^{\otimes N}$ is known \citep{torgersen1991comparison}. 
So far however, the questions of (1) whether mechanisms which are initially \emph{not} Blackwell ranked will eventually \emph{become} Blackwell ranked and (2) which of their properties determine the resulting ranking have not been directly investigated. 

The next result follows from the fact that --under specific preconditions-- composition qualitatively transforms mechanisms towards Gaussians mechanisms (GMs) due to a central limit theorem (CLT)-like phenomenon \citep{dong2022gaussian, sommer2018privacy}. 
Since GMs are always Blackwell ranked (see \cref{lemma:gaussian-trade-off-ordering} in the Appendix), we expect Blackwell dominance to emerge once mechanisms are sufficiently well-approximated by GMs.
We first define:
\begin{equation}
  \eta = \frac{v_1}{\sqrt{v_2 - v_1^2}},
\end{equation}
which plays an important role in the analysis below.
Moreover, in the sequel, $v_1, v_2, v_3$ and $v_4$ will denote the following functionals of $f$:
\begin{align}
    v_1 &= - \int_0^1 \log \left |\frac{\rd}{\rd x}f(x)\right |\rd x, \\
    v_2 &= \int_0^1 \log^2 \left | \frac{\rd}{\rd x}f(x) \right|\rd x, \\
    v_3 &= \int_0^1 \left\lvert\log \left \rvert\frac{\rd}{\rd x}f(x)\right\rvert + v_1 \right\rvert^3 \rd x, \textand\\ 
    v_4 &= \int_0^1 \left |\log \left|\frac{\rd}{\rd x}f(x) \right | \right|^3\rd x.
\end{align}
Intuitively, these represent moments of the PLRV.
\begin{restatable}{lemma}{rankingresolution}\label{thm:ranking-resolution}
  Let $\{\mech_{Ni}:1 \leq i \leq N \}_{N=1}^\infty$ be a triangular array of mechanisms satisfying the following conditions:
  \begin{enumerate}
      \item $\lim\limits_{N\to\infty} \sum_{i=1}^N v_1(f_{Ni}) = K$;
      \item $\lim\limits_{N\to\infty} \max_{1\leq i\leq N} v_1(f_{Ni}) = 0$; 
      \item $\lim\limits_{N\to\infty} \sum_{i=1}^N v_2(f_{Ni}) = s^2$; 
      \item $\lim\limits_{N\to\infty} \sum_{i=1}^N v_4(f_{Ni}) = 0$. 
  \end{enumerate}
  Analogously, define $\{\widetilde\mech_{Ni}: 1 \leq i \leq N \}_{N=1}^\infty$ for constants $\widetilde K, \widetilde s$.
  Then, if $\nicefrac{K}{s}>\nicefrac{\widetilde K}{\widetilde s}$, there exists $N^*$ such that, for all $N\geq N^*$:
  \begin{align}
    \mech_{N1}\otimes \dots \otimes \mech_{NN}\succeq \widetilde\mech_{N1}\otimes \dots \otimes \widetilde\mech_{NN},
  \end{align}
  where $\mech_{N1}\otimes \dots \otimes \mech_{NN}$ denotes $N$-fold mechanism composition and analogously for $\widetilde \mech_{Ni}$.
\end{restatable} 
Our proof strategy relies on first showing that, under the stated preconditions, mechanisms asymptotically converge to Gaussian mechanisms under composition and combining this fact with the property that Gaussian mechanisms are always either equal, or one Blackwell dominates the other.

The conditions above are also used in \citet{dong2022gaussian} to prove the CLT-like convergence of the trade-off functions of composed mechanisms to that of a GM, which we adapt here to show conditions for the emergence of Blackwell dominance between compositions of mechanisms in the limit.
Concretely, $\lbrace \mech_{Ni} \rbrace_{i=1}^{N}$ is a collection of mechanisms calibrated to provide a certain level of privacy after composition, and the mechanisms in the sequence change (become progressively more private) as $N$ grows to $\infty$ to maintain that level of privacy as more mechanisms are composed. 
However, from the more practical standpoint of comparing instances of DP-SGD with different parameters, we are rather interested in the question of approximate Blackwell dominance after a \emph{finite} number of self-compositions of \emph{fixed} parameter mechanisms.
This is shown next.
\begin{restatable}{theorem}{deltadivcompositionbound}\label{thm:delta-div-composition-bound}
Let $\mech,\widetilde\mech$ be two mechanisms with $v_4,\widetilde v_4<\infty$ and denote by $\mech^{\otimes N}, \widetilde\mech^{\otimes \widetilde N}$ their $N$- and $\widetilde N$-fold self-compositions. 
Then, $\nicefrac{N}{\widetilde N} \geq \nicefrac{\widetilde\eta^2}{\eta^2}$ implies:
  \begin{align}
  \deltadiv{\mech^{\otimes N}}{\widetilde\mech^{\otimes \widetilde N}} \leq 0.56\left(
    \frac{\eta^3v_3}{\sqrt{N}v_1^3} + \frac{\widetilde \eta^3\widetilde v_3}{\sqrt{\widetilde N}\widetilde v_1^3}\right) 
  \end{align}
  In particular, if $N=\widetilde N$, $\eta \geq \widetilde \eta$ implies:
  \begin{align}
  \deltadiv{\mech^{\otimes N}}{\widetilde\mech^{\otimes \widetilde N}} \leq \frac{0.56}{\sqrt{N}}\left(
    \frac{\eta^3v_3}{v_1^3} + \frac{\widetilde \eta^3\widetilde v_3}{\widetilde v_1^3}\right). 
  \end{align}
\end{restatable}
The proof relies on the aforementioned Blackwell dominance properties between Gaussian mechanisms combined with the triangle inequality property of the $\Delta$-divergence and a judicious application of the Berry-Esséen-Theorem.

\cref{thm:delta-div-composition-bound} intuitively states that the $\Delta$-divergence will approach zero not asymptotically as in \cref{thm:ranking-resolution}, but within a specific number of update steps and allows for choosing $N, \widetilde N$ differently. 
Seeing as the number of update steps is a crucial hyper-parameter in DP-SGD \citep{de2022unlocking}, this is required for practical usefulness. 
In addition, it pinpoints the exact relationship between the mechanisms ($\nicefrac{\widetilde\eta^2}{\eta^2}$) that determines which mechanism will eventually dominate.
In particular, if $\nicefrac{N}{\widetilde N} \geq \nicefrac{\widetilde\eta^2}{\eta^2}$, then $\deltadiv{\mech^{\otimes N}}{\widetilde\mech^{\otimes \widetilde N}}$ will vanish at least as fast as $\min\{N,\widetilde N\}^{-\nicefrac{1}{2}}$, and if $N=\widetilde N$, then the emergence of Blackwell dominance depends only on the parameters $\eta, \widetilde \eta$, i.e.\@ on the PLRV moments.
Moreover, this result does \emph{not} require scaling the mechanism parameters at every step to prevent them from becoming blatantly non-private, even at very large numbers of compositions.

\section{Experiments}
\mypar{Approximate Comparisons in Practice}
\cref{fig:comparison-demo} demonstrates a \say{canonical} example of an approximate comparison between the GM ($\sigma=1$) and the Laplace mechanism ($b=1$) on a function with unit global sensitivity.
\begin{figure}[ht]
  \centering
  \includegraphics[width=0.9\columnwidth]{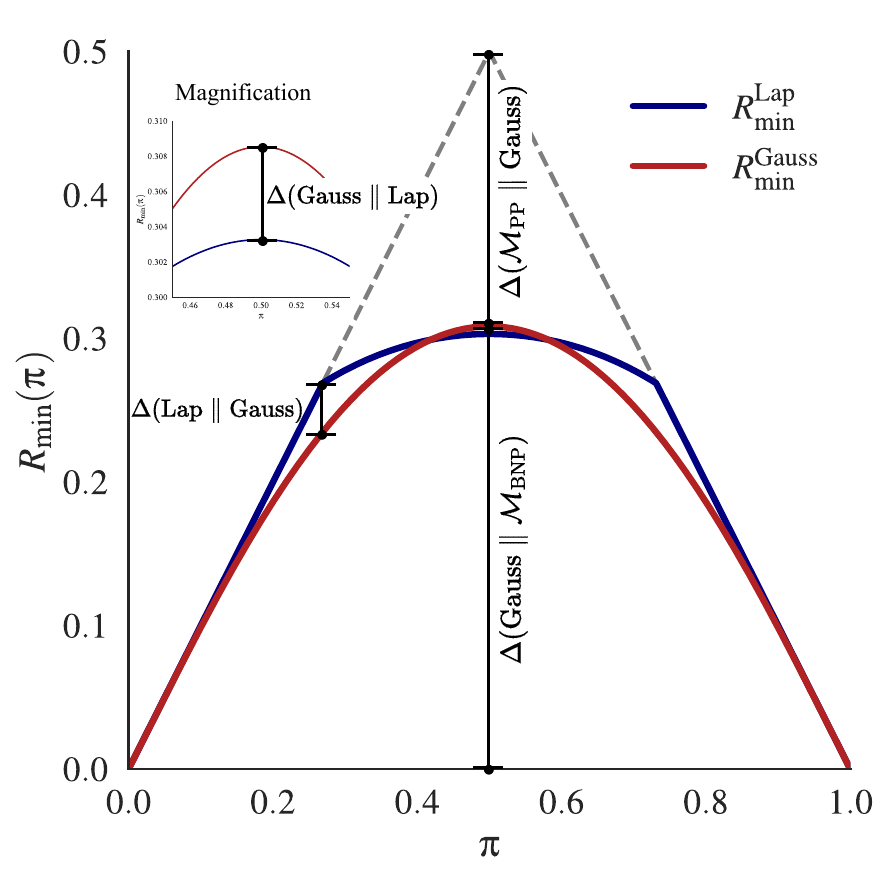}
  \caption{
  Approximate comparison between the Laplace ($b=1$) and the Gaussian ($\sigma=1$) mechanism via their Bayes error representations.
  Note that all divergence values of interest are visually depicted in this representation.
  }
  \label{fig:comparison-demo}
\end{figure}
Observe that the Bayes error functions cross at $\pi\approx0.4$, and that $\deltadiv{\mathrm{Gauss}}{\mathrm{Lap}}=0.005 < \deltadiv{\mathrm{Lap}}{\mathrm{Gauss}}=0.034$, as seen by the length of the black \say{rulers} in the figure. 
Therefore, the worst-case regret in terms of privacy vulnerability of choosing the Laplace mechanism is smaller than for the GM.
Moreover, the Gaussian mechanism offers only marginally stronger protection for a narrow range of $\pi$ around $\nicefrac{1}{2}$ corresponding to the prior of an \say{uninformed} adversary.
This allows for much more granular insights into the mechanisms' privacy properties beyond the folklore statement that \textit{pure DP mechanisms (Laplace) offer \say{stronger privacy} than approximate DP mechanisms (GM)}.

\mypar{Tightness of the Bound in \cref{thm:delta-div-composition-bound}}
To evaluate the bound, we compare two SGMs $\mech, \widetilde \mech$ with $\sigma = 2, \widetilde \sigma = 3$, $p = \widetilde p = 9\cdot 10^{-4}$, $N = 1.4 \cdot 10^{6}$ and $\widetilde N=3.4\cdot 10^{6}$. 
The predicted bound is $\deltadiv{\mech^{\otimes N} }{\widetilde \mech^{\otimes \widetilde N}}< 10^{-3}$, while the empirically computed bound is $8\cdot 10^{-4}$.
The parameter choices in this example mirror those used in \citet{de2022unlocking} for fine-tuning on the JFT-300M dataset to $(8, 5\cdot10^{-7})$-DP, underscoring the applicability of our bound to large-scale ML workflows.

\mypar{Bayesian Mechanism Selection}
An additional benefit of our Bayes error interpretation is that it facilitates principled reasoning about the adversary's auxiliary information.
Recall that $\pi$ expresses the adversary's \say{informedness}, i.e.\@ the strength of their prior belief about the challenge example's membership.
This allows for introducing hierarchical Bayesian modelling techniques to mechanism comparisons by introducing \emph{hyper-priors}, i.e.\@ probability distributions over the adversary's values of $\pi$.
For example, if the defender is very uncertain about the anticipated adversary's prior, they can use an \say{uninformative} hyper-prior such as the Jeffreys prior \cite{jeffreys1946invariant} (here: $\mathrm{Beta}(0.5,0.5)$). 
Alternatively, a more \say{informed adversary} with stronger prior beliefs (i.e.\@ low or high values or $\pi$) could be modelled by e.g.\@ the $\mathrm{UQuadratic}[0,1]$ distribution.
Then, denoting by $\Psi(\pi)$ the hyper-prior, one can obtain the \emph{weighted minimum Bayes error} $\rmin^{\Psi}(\pi) = \rmin(\pi) \Psi(\pi)$. 
Similarly, a weighted $\Delta$-divergence $\Delta^{\Psi}(\mech \parallel \widetilde \mech) = \max_{\pi} ( \rmin^{\Psi}(\pi) - \widetilde R_{\min}^{\Psi}(\pi))$ can be computed, which expresses the excess regret of choosing $\widetilde \mech$ over $\mech$ \emph{modulated by the defender's beliefs} about the adversary's prior.
Incorporating such adversarial priors has recently witnessed growing interest \citep{balle2022reconstructing, jayaraman2021revisiting}.

Our method is a principled probabilistic extension of the recommendation by \citet{carlini2022membership} to choose the mechanism whose trade-off function offers higher Type-II errors at \say{low $\alpha$}.
This recommendation requires a (more or less arbitrary) choice of a \say{low} $\alpha$; as discussed above, no standardised recommendation on this choice exists, leading to poor comparability of results, and potentially skewed reporting.
Moreover, the technique does not take all possible adversaries into account.

These shortcomings are addressed by our proposed technique, as shown in \cref{fig:bayesian}, which compares two SGMs: $\mech$ (blue) and $\widetilde \mech$ (red).
Without any hyper-prior (\cref{fig:bayes-original}), $\deltadiv{\mech}{\widetilde \mech}= 0.01 < 0.02 = \deltadiv{\widetilde \mech}{\mech}$, indicating that choosing $\mech$ is slightly riskier in the worst-case.
Applying the Jeffreys hyper-prior (\cref{fig:bayes-jeffreys}), which expresses a minimal set of assumptions about the adversary, yields $\Delta^{\mathrm{Beta}}(\mech \parallel \widetilde \mech)=0.007 < 0.015 = \Delta^{\mathrm{Beta}}(\widetilde \mech \parallel \mech)$, expectedly not changing the ranking.
However, when the more pessimistic $\mathrm{UQuadratic}[0,1]$ hyper-prior is applied (\cref{fig:bayes-pessimistic}), which models an adversary with strong prior beliefs, we obtain $\Delta^{\mathrm{UQuad}}(\mech \parallel \widetilde \mech)=0.014 > 0 = \Delta^{\mathrm{UQuad}}(\widetilde \mech \parallel \mech)$, indicating that --against an informed adversary-- one would consistently prefer $\mech$.

\begin{figure*}[ht]
  \centering
  \begin{subfigure}[b]{0.33\textwidth}
    \centering
    \includegraphics[width=\linewidth]{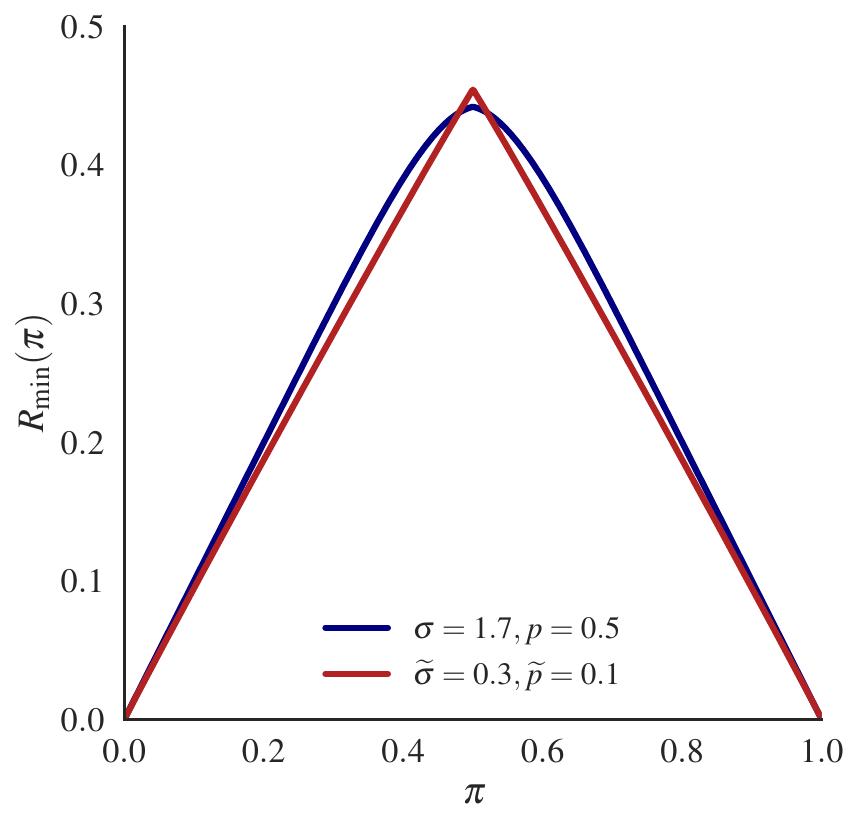}
    \caption{No hyper-prior}
    \label{fig:bayes-original}
  \end{subfigure}
  \begin{subfigure}[b]{0.33\textwidth}
    \centering
    \includegraphics[width=\linewidth]{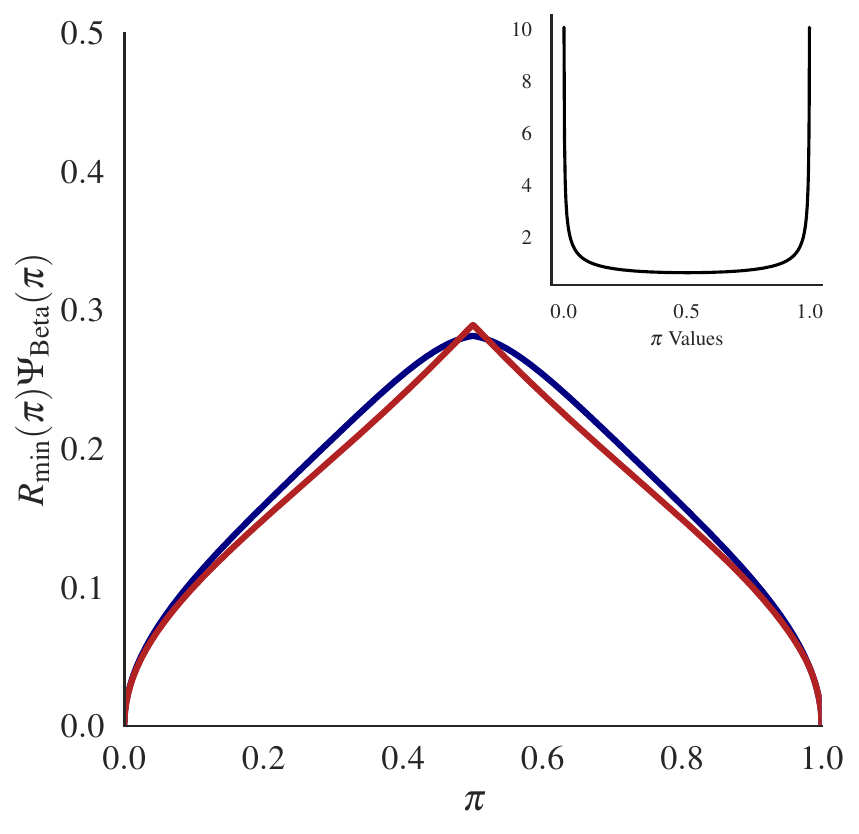}
    \caption{$\mathrm{Beta}(0.5, 0.5)$ hyper-prior}
    \label{fig:bayes-jeffreys}
  \end{subfigure}
  \begin{subfigure}[b]{0.33\textwidth}
    \centering
    \includegraphics[width=\linewidth]{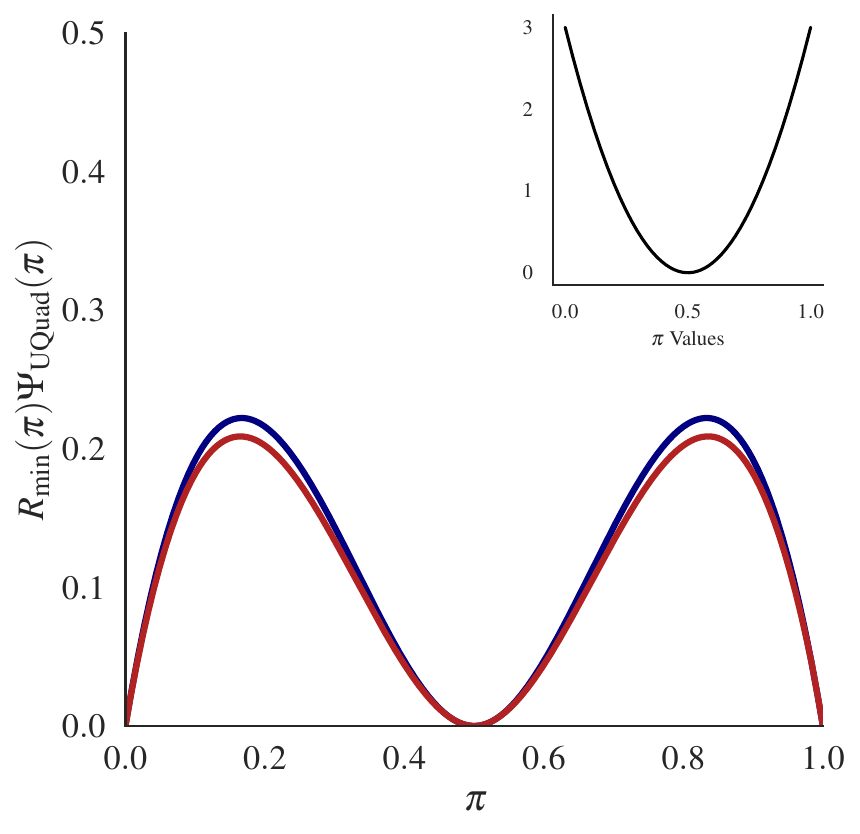}
    \caption{$\mathrm{UQuadratic}[0,1]$ hyper-prior}
    \label{fig:bayes-pessimistic}
  \end{subfigure} 
  \caption{
  Hierarchical Bayesian modelling of adversarial strategies.
  Hyper-prior densities are plotted in the insets.
  }
  \label{fig:bayesian}
\end{figure*}

\mypar{Pareto-Efficient Choice of Noise Multipliers}
Deep learning with DP-SGD presents a trilemma between model accuracy, privacy protection and resource efficiency.
The privacy-accuracy trade-off is well-known in the community, whereas the efficiency trade-off is more apparent when training deep learning models on large-scale datasets.
In the recent work of \citet[Section 5]{de2022unlocking}, the authors posit that there exists an \say{optimal} combination of noise multiplier and number of update steps to achieve the best possible accuracy.
Concretely, the authors calibrate seven CIFAR-10 training runs with different noise multipliers and numbers of steps while fixing the sampling rate to obtain models which all satisfy $(8, 10^{-5})$-DP.
Subsequently, they determine that the \say{optimal} noise multiplier for their application is $\widetilde \sigma = 3.0$, whereas both higher and lower noise multipliers deteriorate the training and validation accuracy.

Here, we re-assess the authors' results using the novel techniques introduced in this paper.
For readability, we will from now equate mechanisms with their noise multipliers, writing e.g. $\deltadiv{\sigma}{\widetilde \sigma=3.0}$ to denote the $\Delta$-divergence from the baseline mechanism $\mech$ with noise multiplier $\sigma=2.0$ and a validation accuracy of $72.6\%$, to $\widetilde \mech$ with $\widetilde \sigma = 3.0$. 
The number of steps and sub-sampling rate are chosen exactly as in \citet[Figure 6]{de2022unlocking}.
\cref{fig:de-et-al-comparison} summarises our observations.
First, note that, even though the mechanisms are all nominally calibrated to $(8, 10^{-5})$-DP, they are not equal in the sense of \cref{lemma:levy}.
This is not unexpected, as it is the same phenomenon observed in \cref{fig:teaser}, and it once again underscores the pitfalls of relying on a single $\epsdel$-pair to calibrate the SGM. 
More interestingly, the $\Delta$-divergence from the baseline increases monotonically with increasing noise multipliers.
This introduces an additional \say{dimension} to the result of \citet{de2022unlocking}:
Choosing $\widetilde \mech$ to have $\widetilde \sigma=3.0$ is not actually an \say{optimal} choice but --at best-- a \emph{Pareto efficient} choice in terms of balancing accuracy and excess vulnerability over $\mech$. 
In particular, choosing $\widetilde\sigma$ to be larger or smaller than $\widetilde\sigma=3.0$ cannot simultaneously increase accuracy and decrease excess vulnerability over $\mech$. 
Thus, in this case, all mechanisms with $\widetilde\sigma>3$ are \emph{Pareto inefficient} choices, since one could simultaneously increase accuracy and decrease excess vulnerability over $\mech$ by choosing $\widetilde\sigma=3.0$.

\begin{figure}[htbp]
  \centering
  \includegraphics[width=\columnwidth]{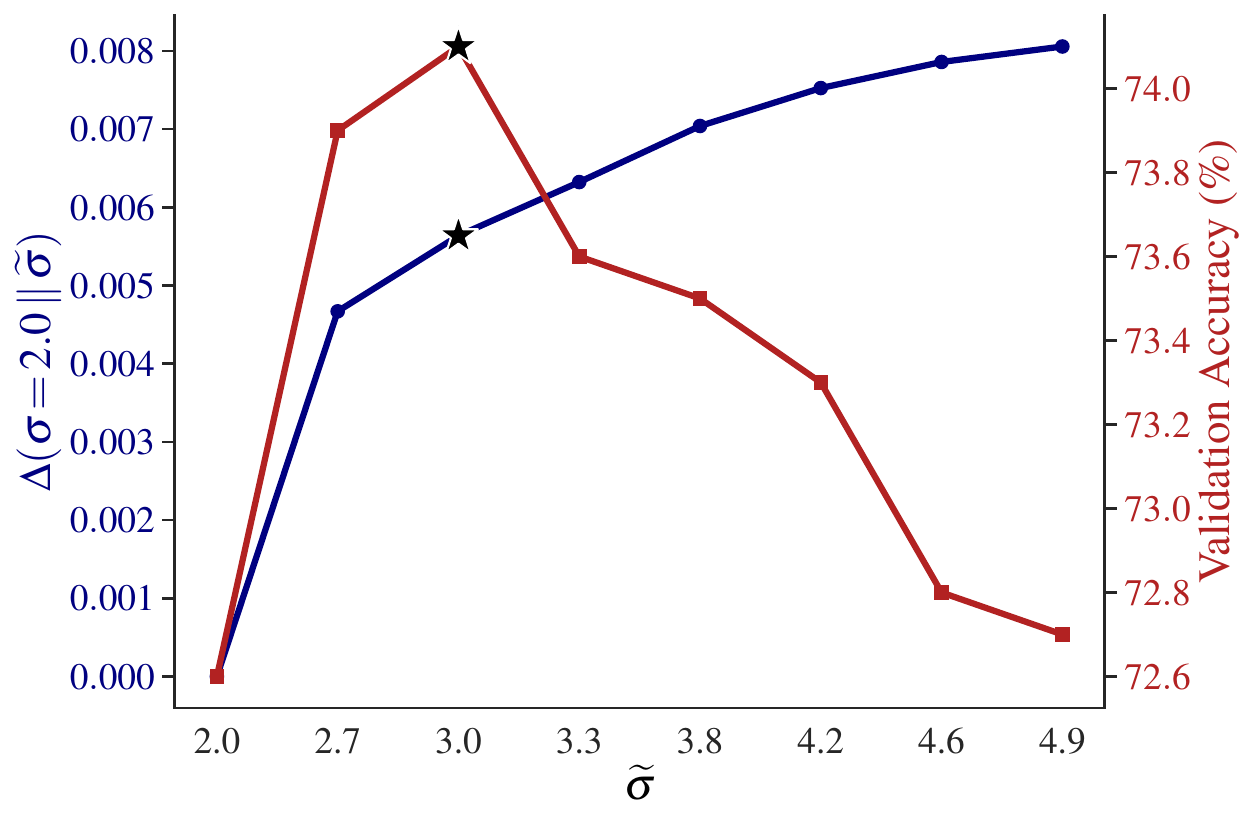}
  \caption{
  Validation accuracy (red) and $\Delta$-divergence values (blue) for mechanisms satisfying $(8, 10^{-5})$-DP.
  The combinations to the right of $\star$ are not Pareto efficient in terms of balancing accuracy and excess vulnerability over $\mech$.
  }
  \label{fig:de-et-al-comparison}
\end{figure}

\mypar{Effect of DP-SGD Parameters on the $\Delta$-Divergence}
To further examine the effect of mechanism parameter choices on the $\Delta$-divergence, \cref{fig:3dplot} investigates switching from a base SGM $\mech$ with $p = 0.01, N = 500$ and $\sigma = 0.54$ to $\widetilde \mech$, where $\widetilde p \in [0.04, 0.9], \widetilde N \in [534, 1500]$ and the resulting $\widetilde\sigma \in [0.55, 21]$.
All mechanisms are calibrated to $(8, 10^{-5})$-DP using the numerical system by \citet{doroshenko2022connect} and the absolute calibration error in terms of $\eps$ is $\leq 0.00042$.
A monotonic increase in the $\Delta$-divergence with the noise multiplier is observed, culminating in a maximum divergence value of around $0.12$. 
In particular, increases in $\widetilde p$ and $\widetilde N$ are associated with an increase in the $\Delta$-divergence. 
Moreover, the $\Delta$-divergence exhibits greater sensitivity to variations in $\widetilde p$ compared to changes in $\widetilde N$. 

\cref{fig:3dplot} suggests that the maximal excess vulnerabilities are realised by large $\widetilde p$ and $\widetilde N$.
This once again highlights not just that these vulnerabilities remain completely undetected when only reporting that the mechanisms \say{satisfy $(8, 10^{-5})$-DP}, but also that the current best practices in selecting SGM parameters for training large-scale ML models with DP, i.e.\@ large sampling rates and many steps \citep{de2022unlocking, berrada2023unlocking} unfortunately correspond to the most vulnerable regime.

\begin{figure}[htbp]
  \centering
  \includegraphics[width=\columnwidth]{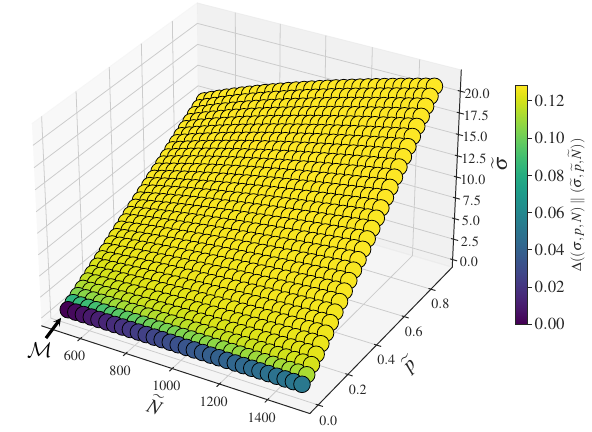}
  \caption{Association between the choice of SGM parameters $(\widetilde \sigma, \widetilde p, \widetilde N)$ and the $\Delta$-divergence compared to a baseline mechanism $\mech$ (marked with an arrow).}
  \label{fig:3dplot}
\end{figure}

\mypar{From $\Delta$-Divergences to Attack Vulnerability}
To provide a practical understanding of what an excess vulnerability of $0.12$ (i.e.\@ the maximum attained in \cref{fig:3dplot}) means in practice, we revisit the example by \citet{hayes2023bounding} discussed in the introduction.
Recall that the authors empirically demonstrated that calibrating different SGMs to a constant $\epsdel$-guarantee while changing the underlying noise multiplier and sampling rate leads to mechanism with disparate vulnerability against data reconstruction attacks. 

Using our newly introduced techniques, we can now formally substantiate this finding, shown in \cref{fig:rero}.
The horizontal axis shows the $\Delta$-divergence value from $\mech$ with $\sigma=0.6, p=0.01$) to a series of mechanisms $\widetilde \mech$ with increasing values of $\widetilde p$ and $\widetilde \sigma$, where all $\widetilde \mech$ are calibrated to $(4, 10^{-5})$-DP as previously described. 
The vertical axis shows the theoretical upper bound on a successful data reconstruction attack against the model (called \emph{Reconstruction Robustness} by \citet{hayes2023bounding}).
We note that these theoretical upper bounds are matched almost exactly by actual attacks, so the bounds are almost tight in practice.
These mechanism settings and resulting reconstruction attack bounds are identical to \citet[Figure 5]{hayes2023bounding}.

Observe that the probability of a successful data reconstruction attack increases almost exactly linearly with the $\Delta$-divergence of the mechanisms from the baseline.
This lends the notion of \say{excess regret} a concrete quantitative interpretation in terms of attack vulnerability, as in this example, an increase of the $\Delta$-divergence from $0$ to $0.12$ corresponds to a $15\%$ (!) vulnerability increase to data reconstruction attacks compared to the baseline.

\begin{figure}[htbp]
  \centering
  \includegraphics[width=\columnwidth]{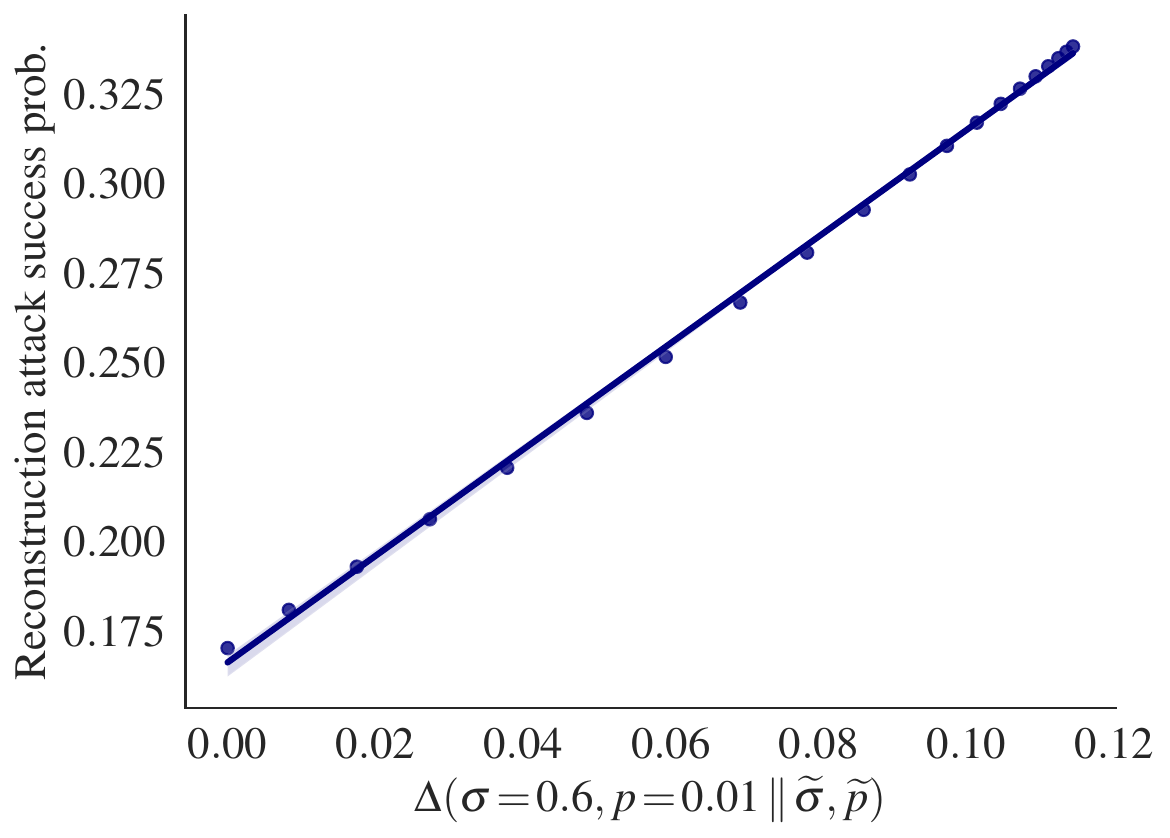}
  \caption{
  $\Delta$-divergence values from a baseline mechanism compared to upper bounds on the success rate of data reconstruction attacks against DP-SGD.
  }
  \label{fig:rero}
\end{figure}

\section{Discussion and Conclusion}
In this work, we established novel mechanism comparison techniques based on the rigorous foundations of the Blackwell theorem.
Our results extend previous works by allowing for principled comparisons between DP mechanisms whose privacy guarantees coincide at the calibration point but differ elsewhere.
Operationally, this enables expressing the regret of switching from one mechanism to another in terms of excess privacy vulnerability in the worst case. 

Our results are supported by a novel Bayesian interpretation, which allows for modelling adversarial auxiliary information.
Such adversarial modelling is currently witnessing increasing interest, as it enables a principled reasoning about adversarial capabilities both in and beyond the worst case \citep{balle2022reconstructing}.
Moreover, our analysis characterises the properties of mechanisms that determine the order of universal Blackwell dominance that inevitably emerges under sufficiently many compositions, which facilitates the application of our results to DP-SGD.
Employing our results to large-scale DP-SGD workflows reveals that calibrating mechanism parameters to attain optimal accuracy must be mindful of associated privacy vulnerabilities, emphasising the risks of the common practice of reporting privacy guarantees in terms of a single $\epsdel$-pair.
Thus, while approximate mechanism comparisons quantify differences between mechanism in terms of privacy vulnerability, we have shown that they can be integrated with considerations of model utility in private ML.
In future work, we aim to additionally incorporate factors such as the cost of training models, into our framework.

In conclusion, the widespread adoption of privacy-enhancing technologies like DP relies heavily on a correct and transparent understanding of privacy guarantees.
Our findings further this understanding, and offer tools to aid informed decision-making in privacy-preserving ML.

\section*{Impact Statement}
We improve the granularity of DP analyses by introducing a novel method to compare privacy guarantees, which can be applied to enhance the security properties of sensitive data processing systems, benefiting individuals.
We foresee no specific negative social consequences of our work.

\section*{Acknowledgements}
GK received support from the German Federal Ministry of Education and Research and the Bavarian State Ministry for Science and the Arts under the Munich Centre for Machine Learning (MCML), from the German Ministry of Education and Research and the the Medical Informatics Initiative as part of the PrivateAIM Project, from the Bavarian Collaborative Research Project PRIPREKI of the Free State of Bavaria Funding Programme \say{Artificial Intelligence -- Data Science}, and from the German Academic Exchange Service (DAAD) under the Kondrad Zuse School of Excellence for Reliable AI (RelAI). 
\bibliography{bibliography}
\bibliographystyle{icml2024}

% APPENDIX
\newpage
\appendix
\onecolumn
\section*{Appendix}

\section{Extended Background}\label{sec:extended-Background}
In this section, we provide an extended introduction to the fundamental concepts used in our work for the purpose of self-containedness and for readers without extensive background knowledge of DP. 

\paragraph{$\epsdel$-DP}
A randomised mechanism $\mech$ satisfies $\epsdel$-DP if, for all adjacent pairs of databases $D$, $D'$ (i.e.\@ differing in the data of a single individual), and all $S \subseteq \text{Range}(\mech)$:
\begin{equation}
    \prob{\mech(D) \in S} \leq \re^{\eps} \prob{\mech(D') \in S} + \delta.
\end{equation}
We will denote adjacent $D, D'$ by $D \simeq D'$.

\paragraph{(Log-) Likelihood Ratios}
The likelihood ratios (LRs) are defined as:
\begin{align}
  \wbx &= \frac{\mathcal{L}(\omega \mid \mech(D'))}{\mathcal{L}(\omega \mid \mech(D))}, \omega \sim \mech(D) \textand \\
  \wby &= \frac{\mathcal{L}(\omega \mid \mech(D'))}{\mathcal{L}(\omega \mid \mech(D))}, \omega \sim \mech(D'),
\end{align}
for arbitrary $D \simeq D'$, where $\mathcal{L}(\omega \mid \cdot)$ denotes the likelihood of $\omega$ and $\sim$ denotes \say{is sampled from}.
Moreover, the log LRs (LLRs) are defined as $X = \log(\wbx)$ and $Y = \log(\wby)$.

The LLRs are customarily called the \emph{privacy loss random variables} (PLRVs), and their densities, denoted $p_X, p_Y$, are called the privacy loss distributions (PLDs).
We will make no other assumptions about $(P,Q)$ other than that they are mutually absolutely continuous for all $D \simeq D'$.
This only excludes mechanisms whose PLDs have non-zero probability mass at $\pm \infty$ e.g.\@ mechanisms which can fail catastrophically, but allows us to study almost all mechanisms commonly used in private statistics/ML.

\paragraph{Hypothesis Testing and $f$-DP}
In the hypothesis testing interpretation \citep{wasserman2010statistical, dong2022gaussian}, a MIA adversary observes a mechanism outcome $\omega$ and establishes the following hypotheses:
\begin{equation}
  H_0: \omega \sim \mech(D) \textand H_1: \omega \sim \mech(D')
\end{equation}
for arbitrary $D \simeq D'$.
$H_0$ is called the null hypothesis and $H_1$ the alternative hypothesis and $H_0$ is tested against $H_1$ using a randomised rejection rule (i.e.\@ test) $\phi: \omega \mapsto \phi(\omega)\in [0,1]$, where $0$ encodes \say{reject $H_0$} and $1$ \say{fail to reject $H_0$}. 
We then denote the Type-I error of $\phi$ by $\aphi = \bE_{\omega \sim \mech(D)}[\phi(\omega)]$ and its Type-II error by $\bphi = 1-\bE_{\omega \sim \mech(D')}[\phi(\omega)]$, where the expectation is over the joint randomness of $\phi$ and $\mech$.

The Neyman-Pearson lemma \cite{neyman1933ix} states that the test with the lowest Type-II error at a given level of Type-I error (called the most powerful test) is constructed by thresholding the (L)LR test statistic; therefore the PLRVs serve as the test statistics for the adversary's hypothesis test.
At a level $\alpha$ fixed by the adversary, the trade-off function $T$ of the most powerful test is given by:
\begin{equation}
  T(\mech(D),\mech(D'))(\alpha) = \inf_{\phi} \lbrace \bphi \mid \aphi \leq \alpha \rbrace.
\end{equation}
$f$-DP \citep{dong2022gaussian} is defined by comparing $T$ to a \say{reference} trade-off function.
Formally, $\mech$ satisfies $f$-DP if, for a trade-off function $f$ and for all $D \simeq D'$:
\begin{equation}
  \forallalpha:\;\sup_{D \simeq D'} T(\mech(D),\mech(D'))(\alpha) \geq f(\alpha).
\end{equation}
Trade-off functions are convex, continuous and weakly decreasing with $f(0)=1$ and $f(1)=0$.
We will, without loss of generality, extend any trade-off function $f$ to $\mathbb{R} \to [0,1]$ and set $f(x)=1, x<0$ and $f(x)=0, x>1$.

\paragraph{Dominating Pairs}
Working with pairs of adjacent databases is not desirable, and not even always feasible when studying general DP mechanisms.
As shown by \citet{zhu2022optimal}, it is instead possible to fully characterise the properties of DP mechanisms by a pair of distributions, called the mechanism's \emph{dominating pair}.
Formally, a pair of distributions $(P,Q)$ is called a dominating pair for mechanism $\mech$ if, for all $\alpha \in [0,1]$ it satisfies:
\begin{equation}
  \sup_{D, D'}T(P,Q)(\alpha) \leq T(\mech(D), \mech(D'))(\alpha).
\end{equation}
In particular, when for all $\alpha \in [0,1]$ it holds that:
\begin{equation}
  \sup_{D, D'}T(P,Q)(\alpha) = T(\mech(D), \mech(D'))(\alpha),
\end{equation}
$(P,Q)$ is called a \emph{tightly dominating pair}.

As noted by \citet{zhu2022optimal}, a tightly dominating pair which encapsulates the worst-case properties of the mechanism, exists or can always be constructed.
Therefore, we will from now on write $\mech:(P,Q)$ to indicate that $(P,Q)$ is a tightly dominating pair of $\mech$, denote the trade-off function corresponding to the most powerful test between $P$ and $Q$ by $f$, its Type-I and Type-II errors by $\alpha$, $\beta(\alpha)$ and the LLRs/PLRVs corresponding to $P$ and $Q$ by $\wbx, X$ and $\wby, Y$.
The trade-off function $f$ can be constructed from $X$ and $Y$ as follows.
Denoting the CDF by $F$:
\begin{equation}\label{eq:tradeoff-construction}
  f(\alpha) = F_{Y}(F_{X}^{-1}(1-\alpha)).
\end{equation}

\paragraph{Privacy Profile}
As shown by \citet{dong2022gaussian, gopi2021numerical}, the privacy profile of $\mech$ can be constructed as:
\begin{equation}
  \deleps = 1 + f^{\ast}(P,Q)(-\re^{\eps}) = \widebar{F}_{Y}(\eps) - \re^{\eps}\widebar{F}_{X}(\eps),
\end{equation}
where $f^{\ast}$ is the convex conjugate and $\widebar{F}$ the survival function.
The privacy profile can also be defined through the \emph{hockey-stick divergence} of order $\re^{\eps}$ of $P$ to $Q$:
\begin{equation}
  \mathsf{H}_{\re^{\eps}}(P \parallel Q) = \int \max \left \lbrace P(x) - \re^{\eps}Q(x), 0 \right \rbrace \rd x = \deleps.
\end{equation}
Note that, for $\eps = 0$, $\mathsf{H}_{1}(P \parallel Q) = \delta(0) = \dtv(P,Q)$,
where 
\begin{equation}
  \dtv(P,Q) = \nicefrac{1}{2}\int \left \lvert P(x) - Q(x) \right \rvert_1 \rd x
\end{equation}
is the total variation distance.

Additionally, the following property holds:
\begin{equation}
  \min_{\alpha\in[0,1]}(\alpha + \beta(\alpha)) = 1-\dtv(P,Q),
\end{equation}
which links the properties of the privacy profile and the trade-off function.
This also allows us to define the \emph{MIA advantage} \citep{yeom2018privacy} of the adversary as follows:
\begin{equation}
  \mathsf{Adv} = 1- \min_{\alpha\in[0,1]} (\alpha + \beta(\alpha)) = \dtv(P,Q).
\end{equation}

\paragraph{Rényi-DP}
Rényi DP (RDP) \citep{mironov2017renyi} is a DP interpretation with beneficial composition properties. 
A mechanism $\mech:(P,Q)$ satisfies $(t, \rho(t))$-RDP if it holds that:
\begin{equation}
  \renyidiv{t}{P}{Q} \leq \rho(t) \; \forall t \geq 1
\end{equation}
for all adjacent $(D, D')$, where $\mathsf{D}_t$ is the Rényi divergence of order $t$.
The conversion between $f$-DP and the privacy profile is exact, but conversions from RDP to either of the aforementioned are not, as RDP lacks a hypothesis testing interpretation \citep{balle2020hypothesis, zhu2022optimal}.

\section{Additional Results}

\subsection{Bayes Error Functions}
Here, we demonstrate the construction of the minimum Bayes error function $\rmin$ from the trade-off function $f$ and \textit{vice versa} using the example of a Gaussian mechanism with $\sigma^2=1$ on a function with unit global sensitivity. 
\cref{fig:bayes-from-tradeoff} shows the construction of $\rmin$ from $f$, while \cref{fig:tradeoff-from-bayes} shows the construction of $f$ from $\rmin$. 
Both directions incur no loss of information, and thus the minimum Bayes error is equivalent to the trade-off function in terms of fully characterising the mechanism.

\begin{figure}[h]
  \centering
  \begin{subfigure}{0.45\linewidth}
    \centering
    \includegraphics[width=\linewidth]{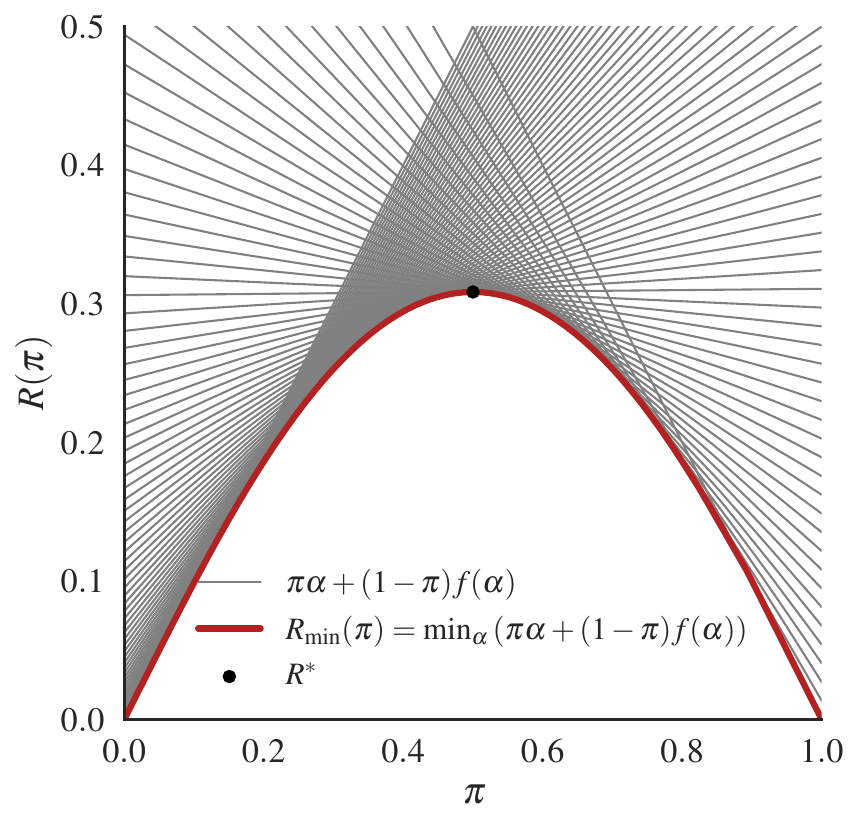}
    \caption{$f \to \rmin$}
    \label{fig:bayes-from-tradeoff}
  \end{subfigure}
  \hfill
  \begin{subfigure}{0.45\linewidth}
    \centering
    \includegraphics[width=\linewidth]{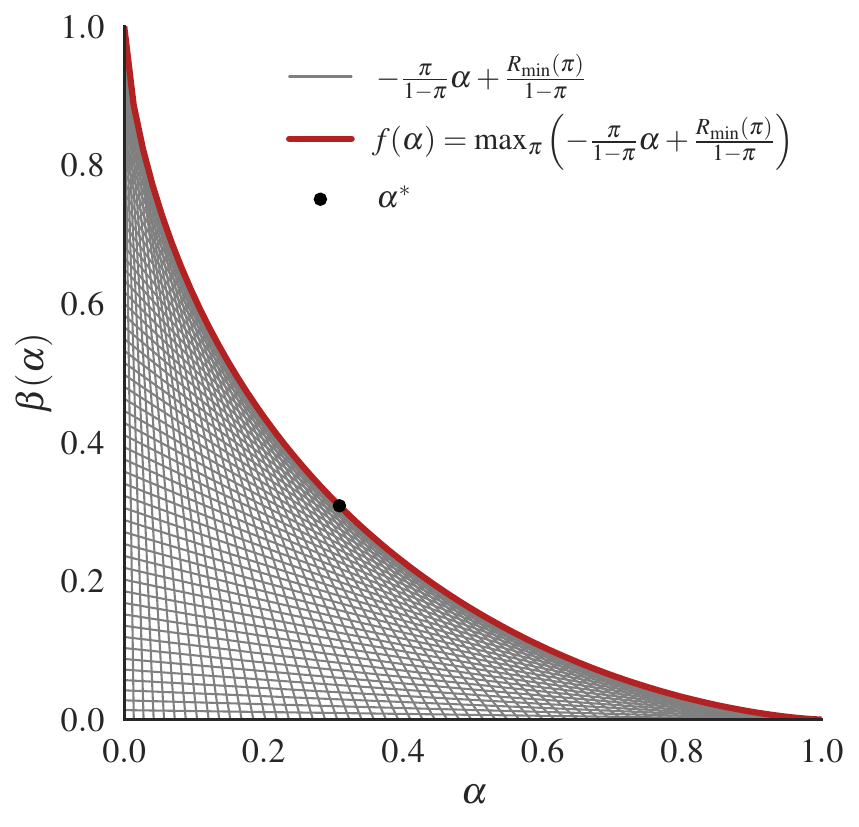}
    \caption{$\rmin \to f$}
    \label{fig:tradeoff-from-bayes}
  \end{subfigure}
  \caption{
  Lossless conversion between $\rmin$ and $f$.
  $\rast$ and $\alpha^{\ast}$ are the minimax Bayes error and fixed point of $f$, respectively.
  }
  \label{fig:tradeoff-bayes}
\end{figure}

\subsection{Interpreting $\Delta^{\leftrightarrow}$ as the Lévy Distance}\label{sec:levy-equivalence}
Following the discussion in \cref{sec:metrising} regarding the conceptual equivalence of the $\Delta$-divergence and the Lévy distance between random variables, we here formally introduce and prove the statement. 
\begin{alemma}
  For any mechanism with trade-off function $f$, $(U,W)$ are a tightly dominating pair, where $U$ is the continuous uniform distribution on the unit interval and $W$ has CDF: $f(1-\alpha)$.
  Moreover, for two mechanisms $\mech :(U, W)$ and $\widetilde \mech:(U, W')$, it holds that $\deltatwosided{\mech}{\widetilde \mech} = \Lambda(W, W')$, where $\Lambda$ denotes the Lévy distance.
\end{alemma}
\begin{proof}
  We will first show that, if $\mech$ is tightly dominated by $(P,Q)$ and has trade-off function $f$, it is also tightly dominated by $(U, W)$.
  From \cref{eq:tradeoff-construction}, $T(U,W)$ is constructed as follows:
  \begin{equation}
    T(U,W)(\alpha) = F_W(F^{-1}_U(1-\alpha)) = F_W(1-\alpha) = f(1-(1-\alpha)) = f(\alpha),
  \end{equation}
  which follows since the inverse CDF (quantile function) of the continuous uniform distribution is the identity function.
  Therefore, $T(U,W)(\alpha) = f(\alpha)$, for all $\alpha\in[0,1]$, hence $(U,W)$ is a tightly dominating pair for $\mech$.
  Next, recall the definition of the Lévy distance:
  \begin{equation}
    \Lambda(W, W') = \inf\left\{\lambda\geq 0 \mid \forall x\in\mathbb{R}:\; F_W(x - \lambda) -\lambda \leq F_{W'}(x) \leq F_W(x + \lambda) + \lambda ) \right\}.
  \end{equation}
  Denoting $f,\widetilde f$ the trade-off functions of $\mech, \widetilde \mech$ respectively and inserting the respective CDFs of $W, W'$, we obtain:
  \begin{align}
    \Lambda(W, W') &= 
    \inf \{\lambda\geq 0\mid \forall \alpha\in\mathbb{R}:\; f(1-(\alpha - \lambda)) -\lambda \leq \widetilde f(1-\alpha) \leq f(1-(\alpha + \lambda)) + \lambda \} \\
    &=\inf \{\lambda\geq 0\mid \forall \alpha\in\mathbb{R}:\; f(1-\alpha + \lambda) -\lambda \leq \widetilde f(1-\alpha) \leq f(1-\alpha -\lambda) + \lambda \}.
  \end{align}
  We can reparameterise the inequality chain from $1-\alpha$ to $\alpha$ to obtain:
  \begin{equation}
    \Lambda(W, W') = \inf \{\lambda\geq 0\mid \forall \alpha\in\mathbb{R}:\; f(\alpha + \lambda) -\lambda \leq \widetilde f(\alpha) \leq f(\alpha -\lambda) + \lambda \}.
  \end{equation}
  Noticing that the result is identical to the definition of $\deltatwosided{\mech}{\widetilde \mech}$ completes the proof.
\end{proof}

\subsection{$\Delta$-Divergence Implementation}\label{sec:divergence-computation}
The following code listing implements the $\Delta$-divergence computation corresponding to the mechanisms in \cref{fig:comparison-demo} in Python.
As seen, the algorithm only requires oracle access to a function implementing the trade-off function of the mechanism.

\begin{lstlisting}[language=Python]
from scipy.stats import norm, laplace
import numpy as np
from functools import partial
from scipy.optimize import minimize_scalar
from multiprocessing import Pool
from os import cpu_count
from typing import Callable, Sequence, Union

def f_gauss(alpha: Union[Sequence[float], float], mu: float) -> float:
    "Gaussian mechanism trade-off function at alpha with parameter mu."
    assert (alpha >= np.zeros_like(alpha)).all() and (
        alpha <= np.ones_like(alpha)
    ).all(), "alpha must be in [0, 1]"
    assert mu >= 0, "mu must be non-negative"
    return norm.cdf(norm.isf(alpha) - mu)

def f_lap(alpha: Union[Sequence[float], float], mu: float) -> float:
    "Laplace mechanism trade-off function at alpha with parameter mu."
    assert (alpha >= np.zeros_like(alpha)).all() and (
        alpha <= np.ones_like(alpha)
    ).all(), "alpha must be in [0, 1]"
    assert mu >= 0, "mu must be non-negative"
    return laplace.cdf(laplace.isf(alpha) - mu)

def _compute_one_rmin(
    pi: float,
    f: Callable[[Union[Sequence[float], float]], float],
) -> float:
    assert 0 <= pi <= 1, "pi must be in [0, 1]"

    def func(alpha: float) -> float:
        assert 0 <= alpha <= 1, "alpha must be in [0, 1]"
        return pi * alpha + (1 - pi) * f(alpha)

    return minimize_scalar(func, bounds=(0, 1)).fun

def rmin(
    *,
    f: Callable[[Union[Sequence[float], float]], float],
    tol: float = 1e-4,
    n_jobs: int = -1,
) -> Union[Sequence[float], float]:
    "Bayes error function corresponding to f computed with tolerance tol."
    assert tol > 0, "tol must be positive"
    assert n_jobs == -1 or n_jobs > 0, "n_jobs must be positive or -1"
    N: int = int(np.ceil(1 / tol))
    pis: Sequence[float] = np.linspace(0, 1, N)
    if n_jobs == -1:
        processes = cpu_count()
    else:
        processes = n_jobs
    with Pool(processes) as pool:
        result = np.array(pool.map(partial(_compute_one_rmin, f=f), pis))
    return result

if __name__ == "__main__":
    tol: float = 1e-4
    mu: float = 1.0
    rmin_lap: Sequence[float] = rmin(f=partial(f_lap, mu=mu), tol=tol, n_jobs=-1)
    rmin_gauss: Sequence[float] = rmin(f=partial(f_gauss, mu=mu), tol=tol, n_jobs=-1)
    divergence_gauss_lap: float = max(rmin_gauss - rmin_lap)
    divergence_lap_gauss: float = max(rmin_lap - rmin_gauss)
    print(f"Delta(Gauss || Lap): {divergence_gauss_lap:.3f}") #prints 0.005
    print(f"Delta(Lap || Gauss): {divergence_lap_gauss:.3f}") #prints 0.034
\end{lstlisting}

\subsection{Proofs}\label{subsec:proofs}
\blackwelltheorem*
\begin{proof}
   For a full proof, see the proof of \cref{thm:approx-comparison}, which recovers \cref{thm:blackwell} for $\fD=0$. 
\end{proof}
\approxcomp*
\begin{proof} 
    \phantom{}\newline
    \mypar{(1)}:
    Suppose $\Delta = \deltadiv{\mech}{\widetilde \mech}\leq \fD$. Since trade-off functions are weakly decreasing, we have:
    \begin{align}
      f(\alpha + \fD) - \fD \leq f(\alpha + \Delta) - \Delta \leq \widetilde f(\alpha).
    \end{align}
    Conversely, if $f(\alpha + \fD) - \fD\leq \widetilde f(\alpha)$, then we have $\deltadiv{\mech}{\widetilde \mech}\leq \fD$ due to the infimum definition of the $\Delta$-Divergence.
    \mypar{(1) $\Rightarrow$ (2)}: 
    Suppose for all $-\infty<\alpha<\infty$, we have $f(\alpha + \fD) -\fD \leq \widetilde f(\alpha)$.
    From \citet{dong2022gaussian}, we know that $\delta(\eps) = 1 + f^{\ast}(-\re^{\eps})$, where:
    \begin{align}
      f^{\ast}(x) = \sup_{-\infty<\alpha<\infty}  (x\alpha - f(\alpha)).
    \end{align}
    denotes the convex conjugate.
    By direct computation of the convex conjugate we obtain:
    \begin{align}
      &\delta(\eps) - 1 = f^*(-\re^\eps) =
      \sup_{-\infty<\alpha<\infty} ( -\re^{\eps}\alpha - (f(\alpha - \fD + \fD) -\fD) - \fD ) \geq 
      \sup_{-\infty<\alpha<\infty} ( -\re^{\eps}\alpha - \widetilde f(\alpha- \fD) -\fD) \\
      &=\sup_{-\infty<\alpha<\infty} (-\re^{\eps}(\alpha + \fD) - \widetilde f(\alpha)) - \fD = 
      \widetilde f^{\ast}(-\re^{\eps}) - \fD \cdot ( 1+ \re^{\eps}) = \widetilde \delta(\eps) - 1 - \fD \cdot (1+\re^{\eps}),
    \end{align}
    which yields the desired inequality.
    
    \mypar{(2) $\Rightarrow$ (1)}: 
    Suppose that, for all $0\leq \eps<\infty$, we have $\delta(\eps) + \fD \cdot(1+e^\eps) \geq \widetilde \delta(\eps)$. 
    Define the function $\widetilde f(\alpha) = \widetilde f(\alpha - \fD) + \fD$. 
    We then have for all $\eps\geq 0$:
    \begin{align}
    &f^*(-\re^\eps) = \delta(\eps)-1 \geq \widetilde \delta(\eps)  - 1 -\fD\cdot(1+\re^\eps) = \widetilde f^*(-\re^\eps) -\fD\cdot(1+\re^\eps) 
    \\
    &= \sup_{-\infty<\alpha<\infty} (-\re^\eps \alpha - \widetilde f(\alpha)) -\fD\cdot(1+\re^\eps) = \sup_{-\infty<\alpha<\infty} (-\re^\eps (\alpha + \fD) - (\widetilde f(\alpha) +\fD)) 
    \\
    &= \sup_{-\infty<\alpha<\infty} (-\re^\eps \alpha - (\widetilde f(\alpha - \fD) +\fD))  = \sup_{-\infty<\alpha<\infty} (-\re^\eps \alpha - \widetilde f(\alpha)) = {\widetilde f}^*(-\re^\eps).
    \end{align}
    This shows $f^*\geq {\widetilde f}^*$, which implies $f\leq \widetilde f$ since the convex conjugate is order-reversing.
    By definition of $\widetilde f$, we showed for all $\alpha$:
    \begin{align}
      f(\alpha) \leq \widetilde f(\alpha - \fD) + \fD.
    \end{align}

    \mypar{(1) $\Rightarrow$ (3)}:
    Suppose for all $\alpha\in[0,1]$ we have $f(\alpha +\fD) -\fD \leq \widetilde f(\alpha)$.
    Let $\alpha\in[0,1]$, such that $\widetilde R_{\min}(\pi)=\pi\alpha + (1-\pi) \widetilde f(\alpha)$. 
    If $\alpha +\fD\in[0,1]$, then we have:
    \begin{align}
      &\rmin(\pi) \leq \pi (\alpha+\fD) + (1-\pi) f(\alpha+\fD) \leq \pi (\alpha+\fD) + (1-\pi) (\widetilde f(\alpha) + \fD)\\
      & = \rmin(\pi) + \fD.
    \end{align}
    In the other, case, we have $\alpha + \fD>1$. But then, $\alpha -\fD\in[0,1]$ since $\alpha\in[0,1]$. 
    Using $f(1)=0=f(\alpha+\fD)$, we also obtain the desired bound:
    \begin{align}
      &\rmin(\pi)\leq \pi + (1-\pi)f(1) \leq \pi(\alpha +\fD) + (1-\pi)f(\alpha+\fD) \leq \pi (\alpha+\fD) + (1-\pi) (\widetilde f(\alpha) + \fD)\\
      & = \rmin(\pi) + \fD.
    \end{align}
    \mypar{(3) $\Rightarrow$ (1)}:
    Suppose $\max_\pi\rmin(\pi) - \widetilde R_{\min}(\pi) \leq \fD$. 
    Let $\alpha\in[0,1]$. If $\alpha +\fD>1$, then trivially $f(\alpha+\fD) -\fD = -\fD \leq 0 = f(\alpha)$ holds. 
    Thus, assume $\alpha +\fD\in[0,1]$. Then, there exists a $\pi\in[0,1]$ such that:
    \begin{align}
    \rmin(\pi) = \pi (\alpha+\fD) + (1-\pi)f(\alpha+\fD).
    \end{align}
    We use the fact that $\widetilde R_{\min}(\pi)\leq \pi \alpha + (1-\pi)\widetilde f(\alpha)$ and obtain:
    \begin{align}
    &\fD \geq \rmin(\pi) - \widetilde R_{\min}(\pi) \geq \pi (\alpha+\fD) + (1-\pi)f(\alpha + \fD) - (\pi\alpha + (1-\pi)\widetilde f(\alpha))\\
    & =\pi \fD + (1-\pi)(f(\alpha +\fD) - \widetilde f(\alpha)).
    \end{align}
    Subtracting $\pi\fD$ from both sides and subsequently dividing by $1-\pi$ yields the desired inequality:
    \begin{align}
      \fD \geq f(\alpha +\fD) - \widetilde f(\alpha).
    \end{align}
\end{proof}

\deltadivbayes*
\begin{proof}
By definition, we have
\begin{equation*}
   \deltadiv{\mech}{\widetilde \mech} = \inf\{ \kappa\geq 0\mid f(\alpha + \kappa) - \kappa \leq \widetilde f(\alpha) \}.
\end{equation*}
Applying clause (3) in \cref{thm:approx-comparison}, we immediately obtain:
\begin{equation*}
   \deltadiv{\mech}{\widetilde \mech} = \inf\{ \kappa\geq 0\mid \max_\pi(\rmin(\pi) - \widetilde R_{\min}(\pi)) \leq \kappa \}.
\end{equation*}
  The $\inf$ is attained at the largest difference in the Bayes error functions, thus:
    \begin{equation}
    \deltadiv{\mech}{\widetilde \mech} = \max_\pi(\rmin(\pi) - \widetilde R_{\min}(\pi)).
  \end{equation}
\end{proof}

\levy*
\begin{proof}
  Let $\fD = \deltadiv{\mech}{\widetilde \mech}$ and $\mathfrak{F} = \deltadiv{\widetilde \mech}{\mech}$, i.e.\@ we have $\mech \succeq_\fD \widetilde \mech$ and $\widetilde \mech \succeq_\mathfrak{F} \mech$. By \cref{thm:approx-comparison}, we have that $f(\alpha + \fD) - \fD \leq \widetilde f(\alpha)$ and $\widetilde f(\alpha) \leq f(\alpha - \mathfrak{F}) + \mathfrak{F}$, for all $\alpha$. Since trade-off functions are weakly decreasing and $\fD,\mathfrak{F}\leq \Delta^\leftrightarrow$, we have:
  \begin{align}
    f(\alpha + \Delta^\leftrightarrow) - \Delta^\leftrightarrow \leq f(\alpha + \fD) - \fD \leq \widetilde f(\alpha) \leq f(\alpha - \mathfrak{F}) + \mathfrak{F} \leq f(\alpha - \Delta^\leftrightarrow) + \Delta^\leftrightarrow.
  \end{align}
\end{proof}

% \equivalence*
% \begin{proof}
% \cref{thm:blackwell} immediately implies:
% \begin{align*}
%   f = \widetilde f \Leftrightarrow \delta = \widetilde \delta \Leftrightarrow \rmin = \widetilde R_{\min}.
% \end{align*}
% Next, note that $\deltatwosided{\mech}{\widetilde \mech} = 0$ is equivalent to $\deltadiv{\mech}{\widetilde \mech}=\deltadiv{\widetilde \mech}{\mech} = 0$. We obtain:
% \begin{align}
%   \deltatwosided{\mech}{\widetilde \mech} = 0 \Leftrightarrow \mech \succeq \widetilde \mech \;\wedge\; \widetilde \mech \succeq \mech \Leftrightarrow f \leq \widetilde f \;\wedge\; \widetilde f \leq f \Leftrightarrow f = \widetilde f.
% \end{align}
% \end{proof}

\pseudometric*
\begin{proof}
    We need to show that $\deltatwosided{\mech}{\widetilde\mech}=0\Leftrightarrow \mech = \widetilde\mech$ and that $\Delta^{\leftrightarrow}$ is symmetric and satisfies the triangle inequality.
  Applying \cref{lemma:deltadiv-via-bayes} we obtain:
  \begin{align}
    &\deltatwosided{\mech}{\widetilde \mech} = \max \left \{ \deltadiv{\mech}{\widetilde \mech}, \deltadiv{\widetilde \mech}{\mech} \right \} \\
    &= 
    \max \left\{ \max_{\pi}\left(\rmin\left(\pi\right) - \widetilde R_{\min}\left(\pi\right)\right),\max_{\pi}\left(\widetilde R_{\min}\left(\pi\right) - \rmin\left(\pi\right)\right)\right\} = \lVert \rmin - \widetilde R_{\min} \rVert_{\infty}.
  \end{align}
  Symmetry, triangle inequality, and $\deltatwosided{\mech}{\mech}=0 \Leftrightarrow \mech = \widetilde\mech$ follow from the fact that $\lVert \cdot \rVert_\infty$ is a norm. 
\end{proof}

\begin{remark}\label{remark:equivalence}
  Note that we introduced the order relation $\succeq$ which is implied by the Blackwell theorem as a \textit{partial order}, and refer to mechanisms as \textit{equal} ($\mech =\widetilde \mech$) if and only if they offer identical privacy guarantees.
  Moreover, we refer to $\Delta^{\leftrightarrow}$ as a \textit{metric} on the space of DP mechanisms.
  This choice is motivated by an operational interpretation: \textit{For all practical intents and purposes, mechanisms which provide identical guarantees are the same mechanism}.
  It is however also possible to subject the aforementioned statements to a more formal order-theoretic treatment, where the symbol \say{$=$} is reserved for objects which satisfy \textit{identity}. 
  Since conferring identical privacy guarantees is not sufficient for being identical, it can be argued that it is more appropriate to refer to distinct mechanisms with identical privacy guarantees as being \textit{equivalent}, and writing $\mech \equiv \widetilde \mech$.
  For example, the mechanisms $\mech: (\mathcal{N}(0, 1), \mathcal{N}(1,1))$ and $\widetilde \mech:(\mathcal{N}(0, 2), \mathcal{N}(2, 2))$ have identical trade-off functions, privacy profiles and Bayes error functions and are thus \textit{equivalent}, but they have different dominating pairs, and are therefore not \textit{identical}. 
  Under this perspective, the order relation $\succeq$ formally loses its antisymmetry property, since $\mech \succeq \widetilde \mech$ and $\widetilde \mech \succeq \mech$ no longer implies that $\mech = \widetilde \mech$ but rather $\mech \equiv \widetilde \mech$, and thus should be referred to as a \textit{preorder}.
  Moreover, since under this treatment, $\deltatwosided{\mech}{\mech}=0$ implies $\Leftrightarrow \mech \equiv \widetilde\mech$ rather than $\Leftrightarrow \mech = \widetilde\mech$, $\Delta^{\leftrightarrow}$ should be referred to as a \textit{pseudometric} (which assigns zero value to \textit{non-identical} (but equivalent) elements).
  We stress that the discussed distinction is largely terminological and does not change any of the results of the paper.
\end{remark}

\extremeness*
\begin{proof}
  ($\mech \succeq \mpp$):
  We have $R_{\min}^{\mathrm{PP}} \geq \rmin$, since, by definition, $R_{\min}^{\mathrm{PP}}(\pi) = \min\{\pi, 1-\pi\}$, and the Bayes error function of any mechanism satisfies $\rmin(\pi) \leq \min\{\pi, 1-\pi\}$, for all $\pi\in[0,1]$. Thus, by \cref{thm:blackwell}, $\mpp$ is Blackwell dominated by any mechanism.
  
  ($\mbnp \succeq \mech$):
  By definition, $R_{\min}^{\mathrm{BNP}}(\pi) = 0$ and thus $R_{\min}^{\mathrm{BNP}}(\pi)\leq \rmin(\pi)$, for all $\pi \in[0,1]$. Thus, by \cref{thm:blackwell}, $\mbnp$ Blackwell dominates any mechanism.
\end{proof}

\perfectprivacy*
\begin{proof}
  We denote the Bayes error functions of $\mech,\mech_{\mathrm{PP}}$ as $\rmin, R_{\min}^{\mathrm{PP}}$ respectively. Note that $R_{\min}^{\mathrm{PP}}(\pi) = \min\{\pi,1-\pi\} \geq \rmin(\pi)$, for all $\pi\in[0,1]$.
  Using \cref{lemma:deltadiv-via-bayes} we obtain:
  \begin{equation}
    \deltadiv{\mpp}{\mech} = \max_\pi ( R_{\min}^{\mathrm{PP}}(\pi) -\rmin(\pi)) = \max_{\pi} \left ( \min\left\{\pi, 1-\pi\right\} - \rmin(\pi) \right).
  \end{equation}  
  Next, note that the maximum of $\min\{\pi, 1-\pi\}$ is at $\pi =\nicefrac{1}{2}$ and that all Bayes error functions are concave by definition and their maximum is also realised at $\pi =\nicefrac{1}{2}$.
  Hence, the largest difference between the perfectly private mechanism and any Bayes risk function must also be at $\pi=\nicefrac{1}{2}$.
  We have:
  \begin{align}
    \deltadiv{\mpp}{\mech} &= \nicefrac{1}{2} - \rmin(\nicefrac{1}{2}) \label{rmin-onehalf}\\ 
    &= \nicefrac{1}{2} - \min_{\alpha\in[0,1]}(\nicefrac{1}{2} \alpha + \nicefrac{1}{2}f(\alpha)) \\
    &= \nicefrac{1}{2}\max_{\alpha\in[0,1]}(1 - \alpha - f(\alpha)) \\
    &= \nicefrac{1}{2}\mathsf{Adv} = \nicefrac{1}{2}\dtv(P,Q) = \nicefrac{1}{2}\delta(0).
  \end{align} 
\end{proof}

\blatantnonprivacy*
\begin{proof}
Since the Bayes risk function of $\mbnp$ is $0$ on the unit interval, the $\Delta$-divergence becomes:
  \begin{equation}
    \deltadiv{\mech}{\mbnp} = \max_\pi(\rmin(\pi) - \rmin^{\mathrm{BNP}}(\pi)) = \max_{\pi \in [0,1]} \left ( \rmin(\pi) - 0 \right ) = \max_{\pi \in [0,1]} \rmin(\pi) = \rast,
  \end{equation}
  where we used \cref{lemma:deltadiv-via-bayes} for the first equality. 
  It remains to show that $R^*=\alpha^*$. Recall that $\rmin$ is concave and symmetric around $\pi=1/2$ and assumes its maximum at $\pi=\nicefrac{1}{2}$. 
  To compute $\rmin(\nicefrac{1}{2})$, we set the following derivative equal to $0$:
  \begin{align}
    \frac{\rd}{\rd \alpha} \left[\nicefrac{1}{2}\pi\alpha + \nicefrac{1}{2}(1-\pi f(\alpha))\right] = 0 \iff \frac{\rd}{\rd \alpha} f(\alpha) = -1 \iff \alpha = f(\alpha).\label{fixed-point equation}
  \end{align}
  The last equivalence follows from the fact that $f$ is a symmetric trade-off function. 
  Denote by $\alpha^*$ the unique point in $[0,1]$ such that $\alpha^* = f(\alpha^*)$. 
  Then, we have:
  \begin{align}
    \deltadiv{\mech}{\mbnp} = \rmin^* = \rmin(\nicefrac{1}{2}) = \nicefrac{1}{2}\alpha^* + \nicefrac{1}{2}f(\alpha^*) = \nicefrac{1}{2}\alpha^* + \nicefrac{1}{2}\alpha^* = \alpha^*.\label{optimal-rmin-at-1/2}
  \end{align}  
\end{proof}

\complementaryerrors*
\begin{proof}
  Since $R_{\min}^{\mathrm{PP}}(\pi) = \min\{ \pi, 1-\pi \}$ which has a maximum at $\pi=\nicefrac{1}{2}$, we have from \cref{rmin-onehalf} that $\deltadiv{\mpp}{\mech} = \nicefrac{1}{2} - \rmin(\nicefrac{1}{2})$.
  Moreover, by \cref{optimal-rmin-at-1/2}, we have $\deltadiv{\mech}{\mbnp} = \rmin(\nicefrac{1}{2})$. 
  Therefore, we obtain:
  \begin{equation}
    \deltadiv{\mpp}{\mech} + \deltadiv{\mech}{\mbnp} = \nicefrac{1}{2} - \rmin(\nicefrac{1}{2}) + \rmin(\nicefrac{1}{2}) = \nicefrac{1}{2}.
  \end{equation}
\end{proof}

Before proceeding with \cref{thm:ranking-resolution} and \cref{thm:delta-div-composition-bound}, we prove the following statements, which will be used below:

\begin{alemma}\label{lemma:gaussian-trade-off-ordering}
  If $G_\mu,G_{\widetilde\mu}$ are two Gaussian trade-off functions with $\mu \leq \widetilde\mu$, then $G_\mu\geq G_{\widetilde\mu}$.
  \begin{proof}
    We will prove that the trade-off function of the Gaussian mechanism is decreasing in $\mu$ for any fixed $\alpha$.
  To show this, we take the first derivative of the trade-off function of the Gaussian mechanism with respect to $\mu$:
  \begin{equation}
    \frac{\partial}{\partial \mu} G_\mu(\alpha) = \frac{\partial}{\partial \mu} \Phi(\Phi^{-1}(1-\alpha) - \mu) = - \frac{\sqrt{2} \re^{- \frac{\left(\mu - \sqrt{2} \operatorname{erfinv}{\left(1 - 2 \alpha \right)}\right)^{2}}{2}}}{2 \sqrt{\pi}},
  \end{equation}
  where $\operatorname{erfinv}$ denotes the inverse error function of the normal distribution.
  Since the exponential is always non-negative, the right hand side is always negative.
  Hence:
  \begin{equation}\label{eq:gaussians-dominate}
    \mu \geq \widetilde{\mu} \iff \forall \alpha\in[0,1]:\; G_{\mu}(\alpha) \leq G_{\widetilde{\mu}}(\alpha).
  \end{equation}
  \end{proof}
\end{alemma}

\begin{alemma}\label{lemma:triangle-inequality}
  Let $\mech_1, \mech_2, \mech_3$ be three mechanisms. Then,
    $\deltadiv{\mech_1}{\mech_3} \leq \deltadiv{\mech_1}{\mech_2} + \deltadiv{\mech_2}{\mech_3}$.
\end{alemma}
\begin{proof}
  Let $\mech_1, \mech_2, \mech_3$ be three mechanisms and $\rmin^1, \rmin^2, \rmin^3$ their respective Bayes error functions.
  Using \cref{lemma:deltadiv-via-bayes}, we have:
  \begin{align}\label{eq:triangle-inequality}
    &\deltadiv{\mech_1}{\mech_3} = 
    \max_{\pi} (\rmin^1(\pi) - \rmin^3(\pi)) 
    \\&= \max_{\pi} (\rmin^1(\pi) -\rmin^2(\pi) + \rmin^2(\pi) - \rmin^3(\pi)) \\
    &\leq \max_{\pi} (\rmin^1(\pi) -\rmin^2(\pi)) + \max_{\pi} (\rmin^2(\pi) - \rmin^3(\pi)) \\
    &\leq \max_{\pi} (\rmin^1(\pi) -\rmin^2(\pi)) + \max_{\pi} (\rmin^2(\pi) - \rmin^3(\pi)) \\
    & =\deltadiv{\mech_1}{\mech_2} + \deltadiv{\mech_2}{\mech_3}.
  \end{align}
\end{proof}

We now proceed with the proofs of \cref{thm:ranking-resolution} and \cref{thm:delta-div-composition-bound} in the main manuscript.
\rankingresolution*
\begin{proof}
    Denote by $ f_{N1} \otimes \dots \otimes f_{NN}$ and $\widetilde f_{N1} \otimes \dots \otimes \widetilde f_{NN}$ the trade-off functions of the compositions $\mech_{N1}\otimes\dots\otimes \mech_{NN}$ and $\widetilde\mech_{N1}\otimes\dots\otimes \widetilde\mech_{NN}$ respectively.
     Next, we apply Theorem 6 in \citet{dong2022gaussian}, which states that these trade-off functions uniformly converge to the Gaussian trade-off functions $G_{2K/s}$ and $G_{2\widetilde K/\widetilde s}$ respectively, i.e.\@
    \begin{align}
      & \lim_{N\to\infty} f_{N1} \otimes \dots \otimes f_{NN} 
      = G_{2K/s},\\
      &\lim_{N\to\infty} \widetilde f_{N1} \otimes \dots \otimes \widetilde f_{NN} = G_{2\widetilde K/\widetilde s}.
    \end{align}
     Suppose $\nicefrac{2K}{s} > \nicefrac{2\widetilde K}{\widetilde s}$ holds.     By \cref{lemma:gaussian-trade-off-ordering}, we then have $G_{2K/s} < G_{2\widetilde K/\widetilde s}$. Moreover, we have:
    \begin{align}
      \lim_{N\to\infty} f_{N1} \otimes \dots \otimes f_{NN} 
      = G_{2K/s}(\alpha) 
      < G_{2\widetilde K/\widetilde s}(\alpha)  = \lim_{N\to\infty} \widetilde f_{N1} \otimes \dots \otimes \widetilde f_{NN}(\alpha) ,
    \end{align}
    where the limits converge uniformly in $\alpha$. In particular, since the limits converge uniformly and are strictly ordered, there must exist $N^*$ such that for all $N\geq N^*$:
    \begin{align}
    f_{N1} \otimes \dots \otimes f_{NN} \leq \widetilde f_{N1} \otimes \dots \otimes \widetilde f_{NN} .
    \end{align}
    This shows if $\nicefrac{2K}{s} > \nicefrac{2\widetilde K}{\widetilde s}$, then there must exist $N^*$ such that for all $N\geq N^*$:
    \begin{equation}
      \mech_{N1}\otimes\dots\otimes \mech_{NN} \succeq \widetilde\mech_{N1}\otimes\dots\otimes \widetilde\mech_{NN}
    \end{equation}
\end{proof}

\deltadivcompositionbound*
\begin{proof}
Assume $\nicefrac{N}{\widetilde N} \geq \nicefrac{\eta}{\widetilde \eta}$. 
Let $\mech_G,\widetilde\mech_{G}$ be two Gaussian mechanisms with trade-off functions $G_{\mu},G_{\widetilde\mu}$ respectively, where:
\begin{align}
  &\mu = \frac{\sqrt{N}2v_1}{\sqrt{v_2 - v_1^2}} = \sqrt{N}2\eta\\
  &\widetilde\mu = \frac{\sqrt{\widetilde N}2\widetilde v_1}{\sqrt{\widetilde v_2 - \widetilde v_1^2}} = \sqrt{\widetilde N} 2 \widetilde\eta.
\end{align}
Between two Gaussian trade-off functions the one with the smaller mean parameter has a larger trade-off function:
\begin{align}
  G_\mu \leq G_{\widetilde\mu}\iff \mu \geq \widetilde\mu \iff \sqrt{N} \eta \geq \sqrt{\widetilde N}\widetilde\eta \iff \nicefrac{N}{\widetilde N} \geq \nicefrac{\eta}{\widetilde \eta}
\end{align}
Since we assumed $\nicefrac{N}{\widetilde N} \geq \nicefrac{\eta}{\widetilde \eta}$, we also have $G_\mu \leq G_{\widetilde\mu}$. In particular, this implies $\deltadiv{\mech_{G}}{\widetilde\mech_{G}}=0$. Next, we apply the triangle inequality from \cref{lemma:triangle-inequality} and obtain:
\begin{align}
  \deltadiv{\mech^{\otimes N}}{\widetilde\mech^{\otimes \widetilde N}} &\leq 
    \deltadiv{\mech^{\otimes N}}{\mech_G} + \deltadiv{\mech_G}{\widetilde\mech^{\otimes \widetilde N}} \\
    & \leq\deltadiv{\mech^{\otimes N}}{\mech_G} + \deltadiv{\mech_G}{\widetilde{\mech}_G} + \deltadiv{\widetilde{\mech}_G}{\widetilde\mech^{\otimes \widetilde N}}\\
    &=\deltadiv{\mech^{\otimes N}}{\mech_G} + \deltadiv{\widetilde{\mech}_G}{\widetilde\mech^{\otimes \widetilde N}}.
\end{align}
To bound the last two summands, we apply Theorem 5 in \cite{dong2022gaussian}, which gives that for all $\alpha\in[0,1]$:
\begin{align}
  &G_\mu(\alpha+\gamma)-\gamma\leq f^{\otimes N}(\alpha)\geq G_\mu(\alpha-\gamma)+\gamma,\\
  &G_{\widetilde\mu}(\alpha+\widetilde\gamma)-\widetilde\gamma\leq \widetilde f^{\otimes \widetilde N}(\alpha) \leq G_{\widetilde\mu}(\alpha-\widetilde\gamma)+\widetilde\gamma,\label{ineq:shift}
\end{align}
where
\begin{align}
  &\gamma = \frac{0.56v_3}{\sqrt{N}(v_2 - v_1^2)^{\nicefrac{3}{2}}}, \\
  & \widetilde\gamma = \frac{0.56\widetilde v_3}{\sqrt{\widetilde N}(\widetilde v_2 - \widetilde v_1^2)^{\nicefrac{3}{2}}}.
\end{align}
In particular, applying a shift by $\widetilde\gamma$ in the last inequality in \cref{ineq:shift} gives for all $\alpha\in[0,1]$:
\begin{align}
  &G_\mu(\alpha+\gamma)-\gamma\leq f^{\otimes N}(\alpha)\;\;\textrm{and}\;\;  \widetilde f^{\otimes \widetilde N}(\alpha + \widetilde\gamma) -\widetilde\gamma \leq G_{\widetilde\mu}(\alpha).
\end{align}
Next, note the definition of the $\Delta$-divergence via the infimum to see that the above implies: 
\begin{align}
  \deltadiv{\mech^{\otimes N}}{\mech_G}\leq \gamma \;\;\textrm{and}\;\; \deltadiv{\widetilde{\mech}_G}{\widetilde\mech^{\otimes \widetilde N}}\leq \widetilde \gamma.
\end{align}
Moreover, we can write $\gamma,\widetilde\gamma$ in terms of $\eta,\widetilde\eta$ respectively:
\begin{align}
  \gamma = \frac{0.56\eta^3v_3}{\sqrt{N}v_1^3}\;\;\textrm{and}\;\; \widetilde\gamma = \frac{0.56\widetilde\eta^3\widetilde v_3}{\sqrt{\widetilde N}\widetilde v_1^3}.
\end{align}
Thus, we have:
\begin{align}
  \deltadiv{\mech^{\otimes N}}{\widetilde\mech^{\otimes \widetilde N}} \leq \gamma + \widetilde \gamma = 0.56\left(\frac{\eta^3v_3}{\sqrt{N}v_1^3} + \frac{\widetilde \eta^3 \widetilde v_3}{\sqrt{\widetilde N}\widetilde v_1^3}\right).
\end{align}
For $N=\widetilde N$ our result above becomes:
\begin{align}
  \deltadiv{\mech^{\otimes N}}{\widetilde\mech^{\otimes \widetilde N}} \leq \frac{0.56}{\sqrt{N}}\left(\frac{\eta^3v_3}{v_1^3} + \frac{\widetilde \eta^3 \widetilde v_3}{\widetilde v_1^3}\right).
\end{align}
\end{proof}

\end{document}